\documentclass[aps,physrev,reprint,superscriptaddress]{revtex4-2}

\usepackage{amsmath}
\usepackage{amsfonts}
\usepackage{amsthm}
\usepackage{amssymb}
\usepackage{graphicx}
\usepackage{dcolumn}
\usepackage{bm}
\usepackage{mathtools}
\usepackage{mathrsfs}
\usepackage{threeparttable}

\theoremstyle{plain}
\newtheorem{theorem}{Theorem}
\newtheorem{lemma}{Lemma}

\begin{document}

\title{Enumeration algorithms for combinatorial problems using Ising machines:\\
       When should we stop exploring energy landscapes?}

\author{Yuta Mizuno}
  \altaffiliation{Present address: Global Research and Development Center for
                  Business by Quantum-AI Technology (G-QuAT),
                  National Institute of Advanced Industrial Science and Technology,
                  Tsukuba, Ibaraki 305-8568, Japan}
  \email{yuta.mizuno@aist.go.jp}
  \email{mizuno@es.hokudai.ac.jp}
  \affiliation{Research Institute for Electronic Science, Hokkaido University,
               Sapporo, Hokkaido 001-0020, Japan}
  \affiliation{Institute for Chemical Reaction Design and Discovery
               (WPI-ICReDD), Hokkaido University, Sapporo, Hokkaido 001-0021, Japan}
  \affiliation{Graduate School of Chemical Sciences and Engineering,
               Hokkaido University, Sapporo, Hokkaido 060-8628, Japan}

\author{Mohammad Ali}
  \affiliation{Graduate School of Chemical Sciences and Engineering,
               Hokkaido University, Sapporo, Hokkaido 060-8628, Japan}
  \affiliation{Statistics Discipline, Khulna University, Khulna 9280, Bangladesh}
\author{Tamiki Komatsuzaki}
  \affiliation{Research Institute for Electronic Science, Hokkaido University,
               Sapporo, Hokkaido 001-0020, Japan}
  \affiliation{Institute for Chemical Reaction Design and Discovery
               (WPI-ICReDD), Hokkaido University, Sapporo, Hokkaido 001-0021, Japan}
  \affiliation{Graduate School of Chemical Sciences and Engineering,
               Hokkaido University, Sapporo, Hokkaido 060-8628, Japan}
  \affiliation{SANKEN, The University of Osaka, Ibaraki, Osaka 567-0047, Japan}

\date{\today}

\begin{abstract}
  Combinatorial problems such as combinatorial optimization and constraint
  satisfaction problems arise in decision-making across various fields of
  science and technology. In real-world applications, when multiple
  optimal or constraint-satisfying solutions exist, enumerating all
  these solutions---rather than finding just one---is often desirable,
  as it provides flexibility in decision-making. However, combinatorial
  problems and their enumeration versions pose significant computational
  challenges due to combinatorial explosion. To address these challenges,
  we propose enumeration algorithms for combinatorial optimization and
  constraint satisfaction problems using Ising machines. Ising machines
  are specialized devices designed to efficiently solve combinatorial
  problems by exploring the energy landscape of an Ising model corresponding
  to the original problem. Ising machines typically sample lower-cost
  solutions with higher probability. Our enumeration algorithms repeatedly
  perform such sampling to collect all desirable solutions. The crux of
  the proposed algorithms lies in their stopping criteria for sampling-based
  energy landscape exploration, which are derived from probability theory.
  In particular, the proposed algorithms have theoretical guarantees
  that the failure probability of enumeration is bounded above
  by a user-specified value, provided that lower-cost solutions
  are sampled more frequently and equal-cost solutions are
  sampled with equal probability. Many physics-based Ising machines
  are expected to (approximately) satisfy these conditions.
  As a demonstration, we applied our algorithm using simulated annealing
  to maximum clique enumeration on random graphs. We found that
  our algorithm enumerates all maximum cliques in large, dense graphs
  faster than a conventional branch-and-bound algorithm specifically
  designed for maximum clique enumeration. These findings underscore
  the effectiveness and potential of our proposed approach.
\end{abstract}

\maketitle

\section{Introduction}

Combinatorial optimization and constraint satisfaction play significant
roles in decision-making across scientific research, industrial development,
and other real-life problem-solving. Combinatorial optimization is the process
of selecting an optimal option, in terms of a specific criterion, from a finite
discrete set of feasible alternatives. In contrast, constraint satisfaction is
the process of finding a feasible solution that satisfies specified constraints
without necessarily optimizing any criterion. Combinatorial problems---which
encompass combinatorial optimization problems and constraint satisfaction
problems---arise in various real-world applications, including
chemistry and materials science \cite{Mizuno2024,Ali2025,Kitai2020},
drug discovery \cite{Sakaguchi2016}, system design \cite{Kirkpatrick1983},
operational scheduling and navigation \cite{Hopfield1985,Rieffel2015,Ohzeki2019},
finance \cite{Rosenberg2016}, and leisure \cite{Mukasa2021}.

Enumerating all optimal or constraint-satisfying solutions is often desirable
in practical applications \cite{Mizuno2024,Ali2025,Eblen2012,Shibukawa2020}.
Desirable solutions of combinatorial problems (i.e., optimal or
constraint-satisfying solutions) are not necessarily unique.
When multiple desirable solutions exist, enumerating all these
solutions---rather than finding just one---provides flexibility
in decision-making. This enables decision-makers to select solutions
that better align with their specific preferences or constraints
not captured in the initial problem modeling.

Despite their practical importance, combinatorial problems and their
enumeration versions pose significant computational challenges.
Many combinatorial problems are known to be NP-hard \cite{Karp1972};
in worst-case scenarios, the computation time to solve such a problem
increases exponentially with the problem size. Moreover, enumerating all
solutions requires more computational effort than finding just
one solution. To address these challenges, we propose enumeration algorithms
for combinatorial problems using Ising machines.

Ising machines are specialized devices designed to efficiently solve
combinatorial problems \cite{Mohseni2022}. The term ``Ising machine" comes
from their specialization in finding the ground states of Ising models
(or spin glass models) in the statistical physics of magnets.
Several seminal studies on computations utilizing Ising models were
published in the 1980s \footnote{
  In particular, the Nobel Prize in Physics 2024 was awarded
  to John J. Hopfield and Geoffrey Hinton for their contributions, including
  the Hopfield network (Hopfield) and the Boltzmann machine (Hinton).
}, including the Hopfield network \cite{Hopfield1982,Hopfield1984}
with its application to combinatorial optimization \cite{Hopfield1985},
the Boltzmann machine \cite{Ackley1985}, and simulated annealing (SA)
\cite{Kirkpatrick1983}. During the same period, early specialized devices
for Ising model simulation were also developed \cite{Pearson1983,Hoogland1983}.
More recently, quantum annealing (QA) was proposed in 1998 \cite{Kadowaki1998}
and physically implemented in 2011 \cite{Johnson2011}. Furthermore,
the quantum approximate optimization algorithm (QAOA) \cite{Farhi2014},
which runs on gate-type quantum computers, typically targets Ising model problems.
The coherent Ising machine (CIM) \cite{Wang2013} is another type of Ising machine,
which utilizes optical phenomena. Currently, various types of Ising machines
are available, as reviewed in \cite{Mohseni2022}.

Many combinatorial problems are efficiently reducible to finding the ground
states of Ising models \cite{Lucas2014,Yarkoni2022}. Ising model problems are
NP-hard \cite{Barahona1982}; thus, any NP problem can be mapped to an Ising model
problem in theory. Furthermore, the real-world applications mentioned above
can also be mapped to Ising model problems \cite{Mizuno2024,Ali2025,Kitai2020,
Sakaguchi2016,Kirkpatrick1983,Hopfield1985,Rieffel2015,Ohzeki2019,
Rosenberg2016,Mukasa2021}. Therefore, Ising machines are widely applicable
to real-world combinatorial problems.

A key feature of Ising machines, especially those based on statistical
\cite{Kirkpatrick1983}, quantum \cite{Kadowaki1998,Johnson2011,Farhi2014},
or optical physics \cite{Wang2013}, is that most of them can be regarded as
samplers from low energy states of Ising models. For instance, SA simulates
a thermal annealing process, where the system temperature gradually decreases.
If the cooling schedule is sufficiently slow, the system is expected to
remain in thermal equilibrium during the annealing process, and thus
the final state distribution is close to the Boltzmann (or Gibbs)
distribution at a low temperature. In fact, the sampling probability
distribution converges to the uniform measure on the ground states, i.e.,
the Boltzmann distribution at absolute zero temperature, for a sufficiently
slow annealing schedule \cite{Geman1984}. Furthermore, quantum and optical Ising
machines, such as (noisy) QA devices \cite{Vuffray2022,Nelson2022,Shibukawa2024},
QAOA \cite{Diez2023,Lotshaw2023}, CIM \cite{Sakaguchi2016},
and quantum bifurcation machine (QbM) \cite{Goto2018}, have theoretical or
empirical evidence that they approximately realize Boltzmann sampling
at a low effective temperature.

We utilize Ising machines as samplers to enumerate all ground states of
Ising models. By repeatedly sampling states using Ising machines, one can
eventually collect all ground states in the limit of infinite samples
\footnote{
  QA devices with transverse-field driving Hamiltonian
  may not be able to identify all ground states in some problems;
  the sampling of some ground state is sometimes significantly
  suppressed \cite{Matsuda2009,Mandra2017,Konz2019}. We do not
  consider the use of such {``unfair"} Ising machines in this article.
}. This raises a fundamental practical question:
\textit{When should we stop sampling?} In this article, we address this question
and derive effective stopping criteria based on probability theory.

The remainder of this article is organized as follows. In Sec.~\ref{sec:models},
we formulate the combinatorial problems and Ising model problems considered
in this article, and define \textit{energy-ordered fair Ising machines} and
\textit{cost-ordered fair samplers} as sampler models. These sampler models
are necessary for deriving appropriate stopping criteria for sampling.
In Sec.~\ref{sec:algorithms}, we propose enumeration algorithms for constraint
satisfaction problems (Algorithm~1) and combinatorial optimization problems
(Algorithm~2). These algorithms have theoretical guarantees that
the failure probability of enumeration is bounded above by a user-specified
value $\epsilon$ when using a cost-ordered fair sampler (or an energy-ordered
fair Ising machine). Detailed theoretical analysis of the failure probability
is provided in Appendix~\ref{appx:theory}. Furthermore, in Sec.~\ref{sec:demo},
we present a numerical demonstration where we applied Algorithm~2 using SA to maximum
clique enumeration on random graphs. Finally, we conclude in Sec.~\ref{sec:conclusion}.

\section{Problem Formulation and Sampler Models} \label{sec:models}

\subsection{Combinatorial Problems and Ising Models}

The combinatorial problems we consider in this article are generally
formulated as
\begin{equation}\label{eq:comb-prob}
\operatorname*{minimize}_{x \in X}\ f(x),
\end{equation}
where $X$ is a finite discrete set of feasible solutions, and
$f\colon X \to \mathbb{R}$ is a cost function to be minimized.
The variable $x$ is typically represented as a binary or integer vector,
and the feasible set $X$ is defined by equality or inequality constraints
on $x$. If $f$ is a constant function, there is no preference between alternatives,
and thus all feasible solutions are desirable solutions; that is, the problem is
a constraint satisfaction problem. Otherwise, it is a (single-objective)
combinatorial optimization problem. 

In many cases, the combinatorial problem defined in Eq.~\eqref{eq:comb-prob}
can be mapped to an Ising model problem:
\begin{equation}
\operatorname*{minimize}_{\bm{\sigma} \in \{-1,1\}^N}
\ H_\mathrm{Ising}(\bm{\sigma}).
\end{equation}
The Ising Hamiltonian $H_\mathrm{Ising}$ is defined as
\begin{equation}
H_\mathrm{Ising}(\bm{\sigma})
= -\sum_{i=1}^{N-1} \sum_{j=i+1}^{N} J_{ij} \sigma_i \sigma_j
  -\sum_{i=1}^{N} h_i \sigma_i.   
\end{equation}
Here, $N$, $\sigma_i$, $J_{ij}$, and $h_i$ denote, respectively,
the number of spin variables, the $i$th spin variable, the interaction
coefficient between two spins $\sigma_i$ and $\sigma_j$, and the local
field interacting with $\sigma_i$. This Ising model should be designed
so that the ground states correspond to the desirable solutions
of the original problem. Standard techniques for mapping combinatorial
problems into Ising model problems can be found in \cite{Lucas2014,Yarkoni2022}.

\subsection{Cost-Ordered Fair Samplers}

To derive appropriate stopping criteria for sampling, we need to specify
a class of samplers (or sampling probability distributions) to be considered.
In this subsection, we define two classes of samplers,
\textit{energy-ordered fair Ising machines} for Ising model problems
and \textit{cost-ordered fair samplers} for general combinatorial problems.
In brief, these sampler models capture the following desirable features
of samplers for optimization: more preferred solutions are sampled more
frequently, and equally preferred solutions are sampled with equal
probability.

First, let us introduce two conditions regarding the sampling probability
distribution of an Ising machine, denoted by $p_\mathrm{Ising}$.
For any two spin configurations $\bm{\sigma}_1$ and $\bm{\sigma}_2$,
\begin{equation}
\begin{cases}
H_\mathrm{Ising}(\bm{\sigma}_1) < H_\mathrm{Ising}(\bm{\sigma}_2)
\Rightarrow
p_\mathrm{Ising}(\bm{\sigma}_1) \ge p_\mathrm{Ising}(\bm{\sigma}_2), \\
H_\mathrm{Ising}(\bm{\sigma}_1) = H_\mathrm{Ising}(\bm{\sigma}_2)
\Rightarrow
p_\mathrm{Ising}(\bm{\sigma}_1) = p_\mathrm{Ising}(\bm{\sigma}_2).
\end{cases}
\end{equation}
The first condition---referred to as the energy-ordered sampling
condition---asserts that a spin configuration with lower energy
is sampled more frequently (or at least with the same frequency)
than a spin configuration with higher energy. In contrast, the second
condition---referred to as the fair sampling condition---states
that two spin configurations with equal energy are sampled with
equal probability. For example, the Boltzmann distribution satisfies
these two conditions. Several Ising machines, such as SA devices \cite{Geman1984},
(noisy) QA devices \cite{Vuffray2022,Nelson2022,Shibukawa2024},
gate-type quantum computers with QAOA \cite{Diez2023,Lotshaw2023},
CIM \cite{Sakaguchi2016}, and QbM \cite{Goto2018}, are known to be
(approximate) Boltzmann samplers for appropriate parameter regimes
(e.g., a sufficiently slow annealing schedule for SA). Therefore,
these Ising machines can be utilized as (approximate) energy-ordered
fair Ising machines. Furthermore, since the energy-ordered and fair
sampling conditions are weaker than the Boltzmann sampling condition,
a broader class of Ising machines is expected to be utilized as
energy-ordered fair Ising machines.

Next, we extend the concept of energy-ordered fair Ising machines
to cost-ordered fair samplers for general combinatorial problems.
We define the cost-ordered and fair sampling conditions regarding
a sampling probability distribution $p$ over feasible solutions of
the combinatorial problem defined in Eq.~\eqref{eq:comb-prob}
as follows. For any two feasible solutions $x_1$ and $x_2$,
\begin{equation}
\begin{cases}
f(x_1) < f(x_2) \Rightarrow p(x_1) \ge p(x_2), \\
f(x_1) = f(x_2) \Rightarrow p(x_1) = p(x_2).
\end{cases}
\label{eq:cost-ordered-fair}
\end{equation}
We define cost-ordered fair samplers as samplers that generate
only feasible solutions and follow the probability distribution $p$ satisfying 
the conditions of Eq.~\eqref{eq:cost-ordered-fair}. Since Ising model problems
are a subset of combinatorial problems, and all spin configurations
are feasible solutions in Ising model problems, energy-ordered fair Ising
machines are also cost-ordered fair samplers for the Ising model problems.

Cost-ordered fair samplers for general combinatorial problems can be
implemented by using energy-ordered fair Ising machines. Typical Ising
formulations of combinatorial problems preserve the order of preference
among solutions \cite{Lucas2014,Yarkoni2022}:
\begin{equation}
\begin{cases}
f(x_1) < f(x_2) \Rightarrow
H_\mathrm{Ising}(\bm{\sigma}_1) < H_\mathrm{Ising}(\bm{\sigma}_2), \\
f(x_1) = f(x_2) \Rightarrow
H_\mathrm{Ising}(\bm{\sigma}_1) = H_\mathrm{Ising}(\bm{\sigma}_2),
\end{cases}
\end{equation}
where $\bm{\sigma}_1$ and $\bm{\sigma}_2$ are the spin configurations
corresponding to feasible solutions $x_1$ and $x_2$, respectively.
Under this condition, the probability distribution over feasible solutions
generated by an energy-ordered fair Ising machine satisfies the cost-ordered
and fair sampling conditions as defined in Eq.~\eqref{eq:cost-ordered-fair}.
Note that not all possible spin configurations that can be sampled by
the Ising machine correspond to feasible solutions of the original (constrained)
combinatorial problem. However, such infeasible solutions can be rejected during
sampling by checking constraint satisfaction, so that the sampler generates
only feasible solutions. (This rejection process will also be illustrated
in the next section.)

In the next section, we will present enumeration algorithms for combinatorial
problems using cost-ordered fair samplers. Although we primarily focus on
cost-ordered fair samplers implemented using energy-ordered fair Ising
machines, our enumeration algorithms can employ a wider class of
stochastic methods and computing devices that satisfy the conditions of
Eq.~\eqref{eq:cost-ordered-fair}. Additionally, the proposed algorithms can
still work effectively in practice, even if the sampler employed does not
strictly meet the conditions of Eq.~\eqref{eq:cost-ordered-fair} 
(see also Sec.~\ref{sec:demo}).

\section{Algorithms} \label{sec:algorithms}

This section describes our proposed algorithms that enumerate all solutions
to (1) a constraint satisfaction problem and (2) a combinatorial optimization
problem using an Ising machine.

\subsection{Preliminaries: Coupon Collector's Problem}

Our enumeration algorithms involve stopping criteria inspired by the coupon
collector's problem, a classic problem in probability theory. This problem
considers the scenario where one needs to collect all distinct items (coupons)
through uniformly random sampling. For example, the number of samples necessary
to collect all $n$ distinct items in a set of cardinality $n$, denoted by
$T^{(n)}_n$, has the following tail distribution:
\begin{equation}
P\left(T^{(n)}_n > \left\lceil n\ln\frac{n}{\epsilon} \right\rceil \right)
< \epsilon,
\end{equation}
where $\lceil\ \rceil$ denotes the ceiling function, and
$\epsilon$ is any positive number less than one.
(See also Lemma~\ref{lemma:tail-distribution-complete}
 in Appendix~\ref{appx:algorithm-1}.)
This inequality suggests that when sampling is stopped at
$\lceil n\ln(n/\epsilon) \rceil$, the failure probability of
collecting all items is less than $\epsilon$. Therefore,
if the number of desirable solutions $n$ were known in advance,
we could simply employ $\lceil n\ln(n/\epsilon) \rceil$ as
the deadline for collecting all desirable solutions. However,
the value of $n$ is unknown in practice. Furthermore,
in a combinatorial optimization problem, nonoptimal solutions are also
sampled in addition to optimal solutions. These challenges demand
an extension of the theory of the coupon collector's problem,
which we address in this article.

In the following two subsections, we present our enumeration algorithms
based on the extended theory of the coupon collector's problem.
Mathematical details are presented in Appendix~\ref{appx:theory}.

\subsection{Enumeration Algorithm for Constraint Satisfaction Problems}

First, we present an enumeration algorithm for constraint satisfaction
problems, referred to as \textit{Algorithm~1} in this article. Algorithm~1
assumes that at least one feasible solution exists and a fair sampler of
feasible solutions is available. Note that for a constraint satisfaction
problem, the cost function $f$ is considered constant; thus
the cost-ordered sampling condition is not required.
The pseudocode is shown in Fig.~\ref{fig:algorithm-1}.

\begin{figure}
    \centering
    \includegraphics[width=\linewidth]{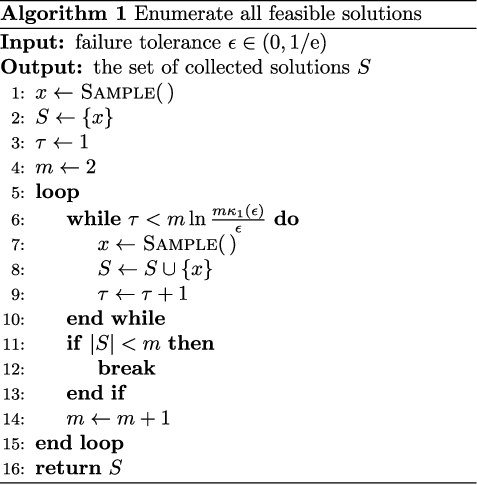}
    \caption{Pseudocode of Algorithm~1. The function \texttt{SAMPLE} is a fair sampler
             of feasible solutions. The definition of $\kappa_1(\epsilon)$,
             which appears in line 6, is provided in Eq.~\eqref{eq:kappa1-def}.
             The failure probability of Algorithm~1 is theoretically guaranteed
             to be less than the user-specified failure tolerance $\epsilon$
             (see Theorem~\ref{theorem:algorithm-1} in Appendix~\ref{appx:algorithm-1}).}
    \label{fig:algorithm-1}
\end{figure}

Algorithm~1 repeatedly samples feasible solutions for a constraint 
satisfaction problem using the function \texttt{SAMPLE}. This function
returns a feasible solution uniformly at random. Such a fair sampler
can be implemented using a fair Ising machine that samples each ground
state of an Ising model with equal probability. In general, it is
easy to check whether a solution satisfies the constraints;
thus, \texttt{SAMPLE} can return only feasible solutions by discarding
infeasible samples generated by the Ising machine.

As the sampling process is repeated and the number of samples
$\tau$ approaches infinity, the set of collected solutions $S$ converges
to the set of all feasible solutions. To stop the sampling process after
a finite number of samples, Algorithm~1 sets the deadline for collecting
$m$ distinct solutions as $\lceil m\ln(m\kappa_1/\epsilon) \rceil$ for
$m = 2, 3, \dots$. Here, $\epsilon$ is a tolerable failure probability for
the enumeration and is required to be less than $1/\mathrm{e}\ (\simeq 0.37)$.
Note that we typically set the tolerable failure probability $\epsilon$
to a much smaller value, such as 0.01 (1\%), and thus this requirement on
$\epsilon$ is not severe. The factor $\kappa_1$ depends on $\epsilon$
but not on the unknown number of desirable solutions to be enumerated.
It is defined as
\begin{equation} \label{eq:kappa1-def}
\kappa_1 \coloneqq
\frac{3^{-2\alpha}}{1-\mathrm{e}^{-\beta}} +
\frac{1}{1-\mathrm{e}^{-\frac{\alpha}{\mathrm{e}-1}}},
\end{equation}
where
\begin{align}
\alpha &\coloneqq \ln\frac{1}{\epsilon} - 1, \\
\beta &\coloneqq
\frac{\frac{1}{\mathrm{e}} + \frac{1}{3}\ln\frac{1}{3}}
     {\frac{1}{\mathrm{e}} - \frac{1}{3}} \alpha.
\end{align}
For instance, $\kappa_1 \simeq 1.14$ when $\epsilon = 0.01$.
Intuitively, the constant $\kappa_1$---which is always larger than one---can
be regarded as a ``correction" factor to the original deadline in the coupon
collector's problem. It slightly extends the deadline to compensate for
the increased error chances caused by the lack of information about
the number of desirable solutions (see Appendix~\ref{appx:algorithm-1}
for detailed discussion). If the number of collected solutions is
fewer than $m$ at $\tau = \lceil m\ln(m\kappa_1/\epsilon) \rceil$,
Algorithm~1 stops the sampling. These specific deadlines ensure that
the failure probability for the enumeration remains below $\epsilon$,
as stated by Theorem~\ref{theorem:algorithm-1} in Appendix~\ref{appx:algorithm-1}.

\begin{figure}
    \centering
    \includegraphics[width=\linewidth]{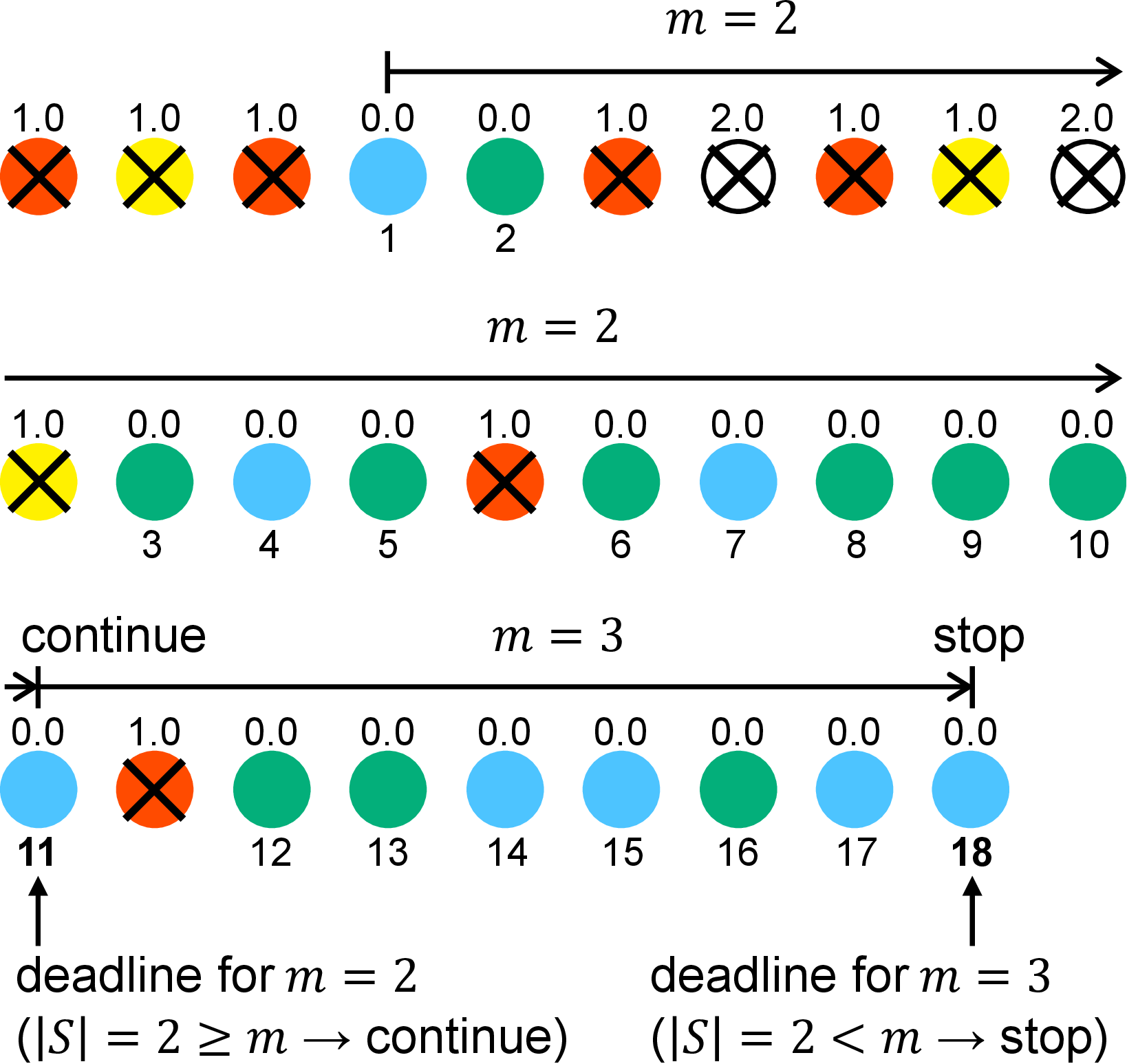}
    \caption{An illustration of a sampling process in Algorithm~1. Each circle
             represents a sample generated by an Ising machine, with different colors
             indicating different solutions. The value above each circle indicates
             the energy of the corresponding solution. In this example, the ground
             state energy is 0.0, and thus light-blue and green circles represent feasible
             solutions (i.e., desirable solutions). In contrast, circles of other colors
             correspond to infeasible solutions, which are discarded during
             the sampling process, as indicated by ``$\times$" marks. The numbers
             below the feasible solutions indicate the sample count $\tau$.
             The deadlines for $m = 2$ and $3$ were calculated as
             $\lceil m\ln(m\kappa_1/\epsilon) \rceil$ with $\epsilon = 0.01$.}
    \label{fig:algorithm-1-example}
\end{figure}

Figure~\ref{fig:algorithm-1-example} illustrates a sampling process in Algorithm~1.
Each circle represents a sample generated by an Ising machine, associated with its
energy value; different colors indicate different solutions. During the sampling process,
infeasible solutions, i.e., samples with energy higher than the ground state energy
of 0.0 (red, yellow, and white in this example) are discarded by checking constraint
satisfaction, as indicated by ``$\times$" marks. This discarding process is a part of
the \texttt{SAMPLE} subroutine in the pseudocode; thus, \texttt{SAMPLE} returns
only feasible solutions without ``$\times$" marks. After sampling the first feasible
solution (light-blue), Algorithm~1 continues sampling until the deadline for
collecting $m = 2$ distinct solutions. This deadline is
$\lceil 2\ln(2\kappa_1/\epsilon) \rceil$, which equals 11 for
$\epsilon = 0.01$. Note that the sample count $\tau$, indicated by
numbers under the circles of feasible solutions, is incremented
only when a feasible solution is sampled. At the deadline $\tau = 11$,
the set of collected solutions $S$ contains two distinct feasible
solutions (light-blue and green). Since the number of collected
solutions $|S|$ equals $m = 2$, Algorithm~1 proceeds to the next phase,
aiming to collect $m = 3$ distinct solutions. The next deadline is
$\lceil 3\ln(3\kappa_1/\epsilon) \rceil$, which equals 18.
However, at the deadline $\tau=18$, the number of collected solutions
$|S|$ is still two, which is fewer than the goal value $m = 3$.
Therefore, Algorithm~1 stops sampling and returns the set $S$
containing the two distinct feasible solutions. In this way,
Algorithm~1 enumerates all feasible solutions, ensuring that
the failure probability for the enumeration is at most $\epsilon$.

\subsection{Enumeration Algorithm for Combinatorial Optimization Problems} \label{sec:algorithm-2}

Next, we present an enumeration algorithm for combinatorial optimization
problems, referred to as \textit{Algorithm~2} in this article. Similar to
Algorithm~1, Algorithm~2 assumes that at least one feasible solution
exists and a cost-ordered fair sampler of feasible solutions is available.
The pseudocode is shown in Fig.~\ref{fig:algorithm-2}.

\begin{figure}
    \centering
    \includegraphics[width=\linewidth]{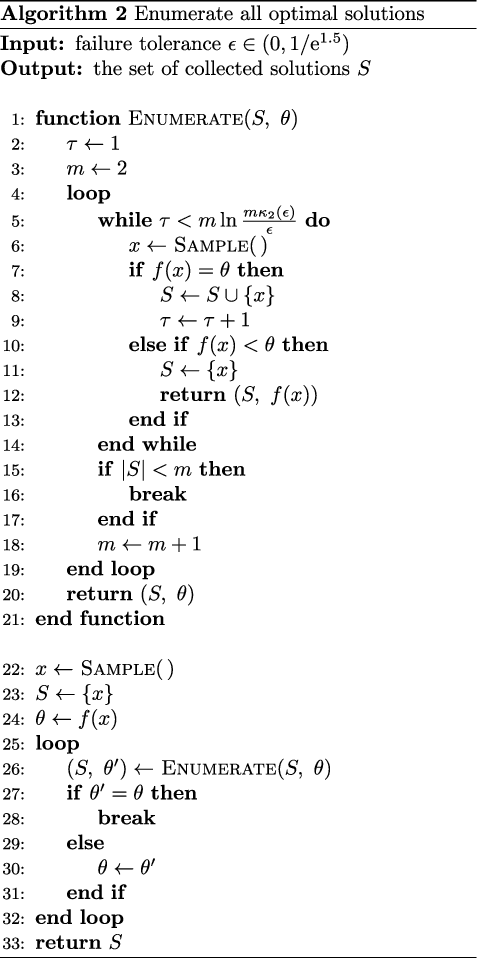}
    \caption{Pseudocode of Algorithm~2. The function \texttt{SAMPLE} is a cost-ordered
             fair sampler of feasible solutions. The definition of $\kappa_2(\epsilon)$,
             which appears in line 5, is provided in Eq.~\eqref{eq:kappa2-def}.
             The failure probability of Algorithm~2 is theoretically guaranteed
             to be less than the user-specified failure tolerance $\epsilon$
             (see Theorem~\ref{theorem:algorithm-2} in Appendix~\ref{appx:algorithm-2}).}
    \label{fig:algorithm-2}
\end{figure}

Enumerating all optimal solutions poses a challenge that does not
arise in enumerating all feasible solutions: it is impossible to
judge whether a sampled solution is optimal without knowing
the minimum cost value in advance. Therefore, Algorithm~2 collects
current best solutions as \textit{provisional} optimal solutions during sampling.
If a solution better than the provisional optimal solutions is sampled,
the algorithm discards the already collected solutions and continues to
collect new provisional optimal solutions.

Specifically, Algorithm~2 holds the minimum cost value among already
sampled solutions as $\theta$. To collect provisional optimal
solutions with cost $\theta$, the algorithm uses the subroutine
\texttt{ENUMERATE} (see lines 1--21 in Fig.~\ref{fig:algorithm-2}).
This subroutine is a modified version of Algorithm~1, which aims to
enumerate all feasible solutions with cost $\theta$. However,
if a solution with cost lower than $\theta$ is sampled
during enumeration, this subroutine stops collecting solutions
and resets the set of collected solutions $S$ so that it contains
only the newly-found better solution. In either case, the subroutine returns $S$
and the current minimum cost value. If the \texttt{ENUMERATE} subroutine stops
enumeration without sampling a better solution (i.e., the current minimum cost
value does not change), Algorithm~2 halts and returns $S$.

The deadline for collecting $m$ distinct solutions employed in
the \texttt{ENUMERATE} subroutine depends on $\kappa_2$, instead of $\kappa_1$
used in Algorithm~1. The constant $\kappa_2$ is defined as
\begin{equation} \label{eq:kappa2-def}
\kappa_2 \coloneqq
\frac{4^\alpha}{1-\mathrm{e}^{-\beta}}
\left(\zeta(2\alpha) - \sum_{k=1}^5 \frac{1}{k^{2\alpha}}\right) +
\frac{2-\mathrm{e}^{-\frac{\alpha}{\mathrm{e}-1}}}
     {\left(1-\mathrm{e}^{-\frac{\alpha}{\mathrm{e}-1}}\right)^2},
\end{equation}
where $\alpha$ and $\beta$ are the same as those used in
Eq.~\eqref{eq:kappa1-def}, and $\zeta$ denotes the Riemann zeta function.
For instance, when $\epsilon=0.01$, $\kappa_2 \simeq 2.44$.
To ensure the convergence of the Riemann zeta function,
$2\alpha\ [= 2\ln(1/\epsilon)-2]$ must be greater than 1,
which implies that $\epsilon < 1/\mathrm{e}^{1.5}\ (\simeq 0.22)$.
Note that this upper limit for allowable $\epsilon$ value is moderate,
as we typically set the tolerable failure probability to a much smaller
value, such as 0.01 (1\%). This specific design of $\kappa_2$ ensures that
the failure probability for the enumeration remains below $\epsilon$,
as stated by Theorem~\ref{theorem:algorithm-2} in Appendix~\ref{appx:algorithm-2}.

The deadlines used in Algorithm~2 are larger than those in Algorithm~1.
This is because $\kappa_2$ is always larger than $\kappa_1$, in order to
compensate for the increased error chances caused by the lack of
information about the true minimum cost (see the end of
Appendix~\ref{appx:algorithm-2} for detailed discussion).

\begin{figure}
    \centering
    \includegraphics[width=\linewidth]{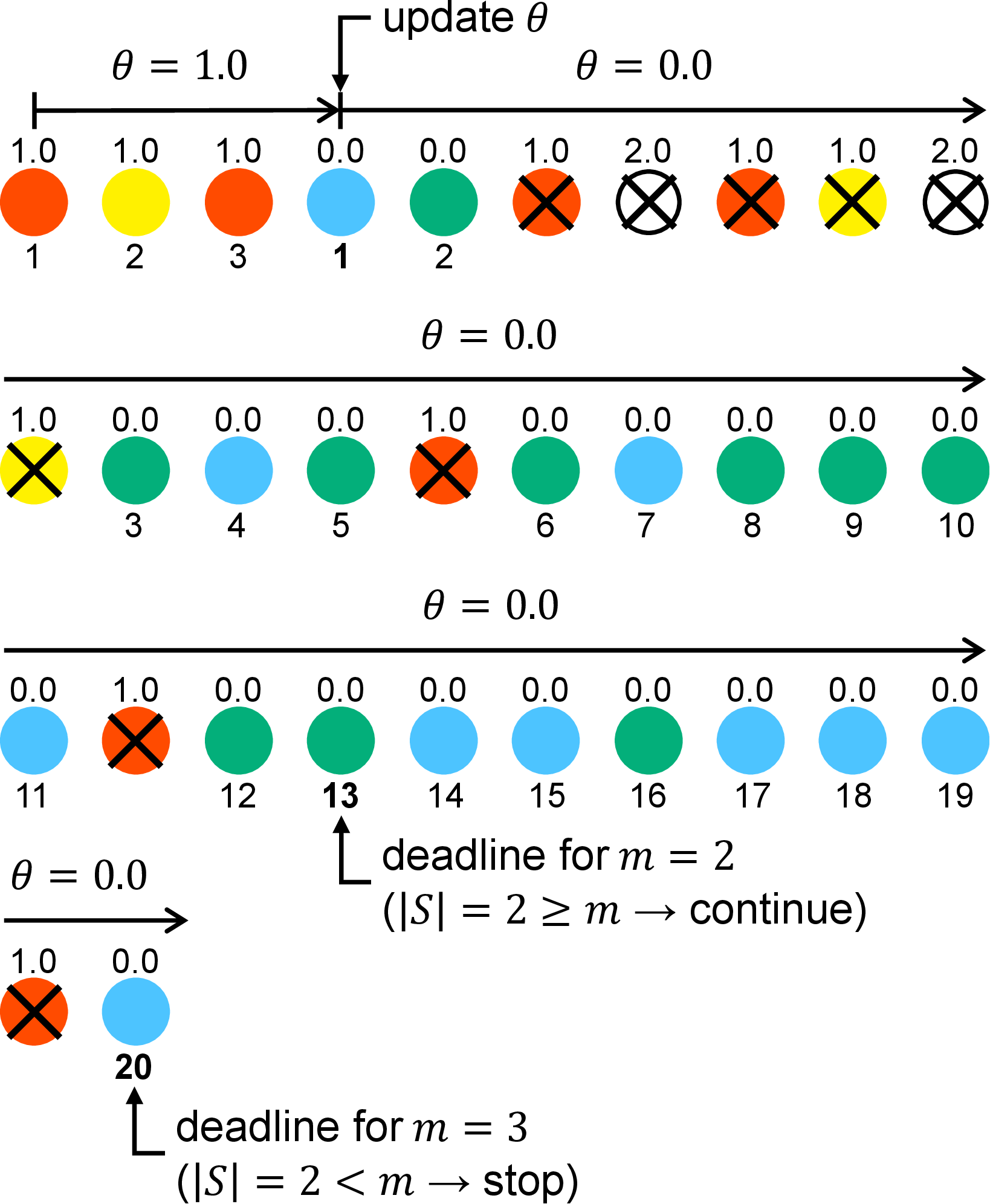}
    \caption{An illustration of a sampling process in Algorithm~2. Each circle
             represents a sample generated by an Ising machine, with different colors
             indicating different solutions. The ``$\times$" marks indicate the rejection
             of the samples. The value above each circle indicates the energy of
             the corresponding solution. Here, for simplicity, all samples are assumed
             to be feasible, and the energy of each feasible solution is set to
             its cost value in the original problem. In this example, the ground state
             energy is 0.0, and thus light-blue and green circles represent the optimal
             solutions. The numbers below the accepted solutions indicate the sample
             count $\tau$. Note that the sample count is reset when the current minimum
             cost value $\theta$ is updated. The deadlines for $m = 2$ and $3$ were
             calculated as $\lceil m\ln(m\kappa_2/\epsilon) \rceil$ with $\epsilon = 0.01$.
             The sample sequence is the same as that in Fig.~\ref{fig:algorithm-1-example}.}
    \label{fig:algorithm-2-example}
\end{figure}

Figure~\ref{fig:algorithm-2-example} illustrates a sampling process in Algorithm~2.
Unlike Fig.~\ref{fig:algorithm-1-example}, here all samples are assumed to be feasible
for simplicity; even in this setting, the desirable solutions to be enumerated remain
the two lowest energy solutions, light-blue and green. In this example, Algorithm~2 does
not reject the first sample (red) that is not a true optimal solution. Instead,
the algorithm collects it as a provisional optimal solution and sets $\theta$
to its cost value 1.0. At this first stage, the algorithm aims to collect
all provisional optimal solutions with cost 1.0 (red and yellow).
However, before the deadline for collecting $m=2$ distinct solutions
with cost 1.0, a better solution (light-blue) is sampled.
Thus, the algorithm updates $\theta$ to 0.0 and resets $S$ and the sample
count $\tau$. At this second stage, the algorithm aims to collect all
provisional optimal solutions with cost 0.0 (light-blue and green).
Solutions with any cost value exceeding the new threshold $\theta=0.0$ are
rejected during sampling, as indicated by the ``$\times$" marks.
The algorithm continues the enumeration until the deadline for
$m=3$ without updating $\theta$. Since the number of collected solutions
with cost 0.0 is two, the algorithm halts at this deadline,
returning $S$ that contains the two collected solutions. In this way,
Algorithm~2 enumerates all optimal solutions by attempting to enumerate
provisional best solutions and updating the current minimum cost
value $\theta$, ensuring that the failure probability
for the enumeration is at most $\epsilon$.

\subsection{Computational Complexity} \label{sec:complexity}

Before concluding this section, we discuss the computational complexity
of the proposed enumeration algorithms. Both Algorithm~1 and Algorithm~2
require $\lceil (n+1)\ln[(n+1)\kappa/\epsilon] \rceil$ samples of desirable
solutions to ensure successful collection of all $n$ desirable solutions,
where $\kappa$ is either $\kappa_1$ or $\kappa_2$.
(Note that the algorithms stop at the deadline for collecting $n+1$
distinct desirable solutions in successful cases.) On the other hand,
the expected time to sample a desirable solution can be estimated
by $\mathcal{T}_\mathrm{sample}/p_\mathrm{desirable}$. Here,
$\mathcal{T}_\mathrm{sample}$ denotes the time to sample a feasible solution
(including both desirable and undesirable ones) using a cost-ordered fair sampler,
and $p_\mathrm{desirable}$ represents the probability that the sampler generates
a desirable solution, i.e.,
$p_\mathrm{desirable} =
 \sum_{x \in \operatorname{argmin}_{x^\prime} f(x^\prime)} p(x)$
[$p(x)$: the cost-ordered fair sampling probability of $x$].
Combining these estimates, we obtain the expected computation
time of the proposed algorithms in cases where enumeration succeeds as
\begin{equation}
\left\lceil (n+1)\ln\frac{(n+1)\kappa}{\epsilon} \right\rceil
\times \frac{\mathcal{T}_\mathrm{sample}}{p_\mathrm{desirable}}.
\end{equation}
Although the first factor does not directly depend on the problem size
(e.g., the number of variables), the second factor may increase exponentially
with the problem size for NP-hard problems. Therefore, the computation time
is dominated primarily by the second factor, i.e., the sampling performance of
the cost-ordered fair sampler (or the Ising machine employed). Moreover,
the number of desirable solutions $n$ may also increase exponentially
with the problem size in worst-case scenarios.

An enumeration algorithm for constraint satisfaction problems utilizing
a fair Ising machine was previously proposed by Kumar et al. \cite{Kumar2020}
and later improved by Mizuno and Komatsuzaki \cite{Mizuno2021}. Their algorithm
is the direct ancestor of Algorithm~1 proposed in this article. In their
algorithm, the deadlines for collecting $m$ distinct solutions are set at
large intervals (e.g., $m = 2, 2^2, \cdots, 2^N$, where $N$ is the number
of spin variables). This leads to additional overhead in the number of
samples required. For instance, when $n=20$, the required number of samples
is $\lceil 32\ln(32\kappa/\epsilon) \rceil$ in their algorithm.
In contrast, our Algorithm~1 requires a much smaller number of samples,
$\lceil 21\ln(21\kappa_1/\epsilon) \rceil$, because in Algorithm~1,
the deadlines are set at every integer value of $m$. Furthermore,
the factor $\kappa$ in the previous algorithm \cite{Mizuno2021} is
typically proportional to $N$, while $\kappa_1$ used in our Algorithm~1
is independent of $N$. This improvement in the computational complexity
of Algorithm~1 results from the careful analysis of the failure probability,
which is detailed in Appendix~\ref{appx:algorithm-1}.

\section{Numerical Demonstration} \label{sec:demo}

This section presents a numerical demonstration of Algorithm~2,
the enumeration algorithm for combinatorial optimization problems
proposed in Sec.~\ref{sec:algorithm-2}. As discussed in Sec.~\ref{sec:complexity},
the actual computation time of the algorithm depends on the performance
of an Ising machine employed. Furthermore, although the algorithm has
the theoretical guarantee of its success rate under the cost-ordered fair
sampling model, the success rate could be different from the theoretical
expectation due to deviations in the actual sampling probability from
the theoretical model. Therefore, we evaluate the actual computation
time and success rate of Algorithm~2 for the maximum clique problem
\cite{Wu2015}, a textbook example of combinatorial optimization.

\subsection{Maximum Clique Problem}

A clique in an undirected graph $G$ is a subgraph in which every two distinct
vertices are adjacent in $G$. Finding a maximum clique, i.e., a clique
with the largest number of vertices, is a well-known NP-hard combinatorial
optimization problem \cite{Karp1972,Lucas2014,Wu2015}. The maximum clique
problem has a wide range of real-world applications from chemoinformatics
to social network analysis \cite{Wu2015}. In particular, enumerating all
maximum cliques is desirable in applications to chemoinformatics \cite{Ali2025}
and bioinformatics \cite{Eblen2012}.

The maximum clique problem on graph $G$ can be formulated as:
\begin{equation}
\begin{alignedat}{2}
&\operatorname*{maximize}_{\bm{x} \in \{0, 1\}^{|V_G|}}& \quad
&\sum_{v \in V_G} x_v, \\
&\operatorname{subject\ to}& \quad
&\forall \{u, v\} \in \overline{E}_G,\; x_u x_v = 0,
\end{alignedat}
\end{equation}
where $V_G$ and $\overline{E}_G$ denote the vertex set and the complementary
edge set (i.e., the set of nonadjacent vertex pairs) of $G$, respectively.
The symbol $\bm{x}$ collectively denotes the binary variables
$\{x_v\}_{v \in V_G}$ and represents a subset of vertices in $G$;
each variable $x_v$ indicates whether the vertex $v$ is included
in the subset ($x_v = 1$) or not ($x_v = 0$). The constraints ensure
that the vertex subset does not include any nonadjacent vertex pairs.
In other words, these constraints exclude vertex subsets that do not
form a clique. Under the clique constraints, the objective is to maximize
the number of vertices included in a clique, which equals $\sum_{v \in V_G} x_v$.

Alternatively, the maximum clique problem can be formulated as
a quadratic unconstrained binary optimization (QUBO) problem:
\begin{equation} \label{eq:max-clique-qubo}
\operatorname*{minimize}_{\bm{x} \in \{0, 1\}^{|V_G|}}
\quad -\sum_{v \in V_G} x_v + A \sum_{\{u, v\} \in \overline{E}_G} x_u x_v,
\end{equation}
where $A$ is a positive constant that controls the penalty for violating
the clique constraints. If $A$ is greater than one, the optimal solutions of
this QUBO formulation are exactly the same as those of the original formulation
\cite{Lucas2014}. By converting the binary variables $x_v$ to spin variables
$\sigma_v\ (\coloneqq 1 - 2x_v)$, the QUBO problem becomes equivalent to
finding the ground state(s) of an Ising model.

\subsection{Computation Methods}

We generated Erdős-Rényi random graphs \cite{Erdos1959} to create benchmark
problems. The number of vertices of each graph $G$, denoted by $|V_G|$,
was randomly selected from the range of 10 to 500. The number of edges
was determined to achieve an approximate graph density $D$, calculated
as $\binom{|V_G|}{2}D$ rounded to the nearest integer. The graph density
parameter $D$ was set to 0.25, 0.5, and 0.75. For each value of $D$,
100 random graphs were generated.

We solved the maximum clique problem on each graph using Algorithm~2
with simulated annealing (SA). The algorithm was implemented in Python,
employing the \texttt{SimulatedAnnealingSampler} from the D-Wave Ocean Software
\cite{DWaveOcean}. The tolerable failure probability $\epsilon$ was set
to 0.01. The penalty strength $A$ in Eq.~\eqref{eq:max-clique-qubo} was set
to a moderate value of 2, as an excessively large penalty strength
may deteriorate the performance of SA. Additionally, we used default
parameters for the SA function. Each randomly generated problem was
solved 100 times using the proposed algorithm, where each run used
a different random seed for SA.

We also solved the same problems using a conventional branch-and-bound
algorithm as a reference. This branch-and-bound algorithm is based on
the Bron--Kerbosch algorithm \cite{Bron1973} with a pivoting technique proposed
by Tomita, Tanaka, and Takahashi \cite{Tomita2006}. This standard algorithm is
implemented in the NetworkX package \cite{NetworkX}, called \texttt{find\_cliques}.
We further modified this \texttt{find\_cliques} function by incorporating a basic
bounding condition for efficient maximum clique search proposed by Carraghan
and Pardalos \cite{Carraghan1990}. This algorithm is exact, i.e., it enumerates
all maximum cliques with 100\% success probability, albeit at the cost of long
computation time, especially for large and dense graphs, as shown below.

All computations were performed on a Linux machine equipped with two Intel
Xeon Platinum 8360Y processors (2.40 GHz, 36 cores each).

\subsection{Results and Discussion}

\subsubsection{Computation time}

First, we compare the computation time of our proposed algorithm
(Algorithm~2 using SA) against that of the conventional algorithm
(Bron--Kerbosch combined with the enhancements by Tomita--Tanaka--Takahashi
and Carraghan--Pardalos). The results are shown in Fig.~\ref{fig:computation-time}
and Table~\ref{tbl:computation-time-scaling}.

\begin{figure*}
    \centering
    \includegraphics[width=\linewidth]{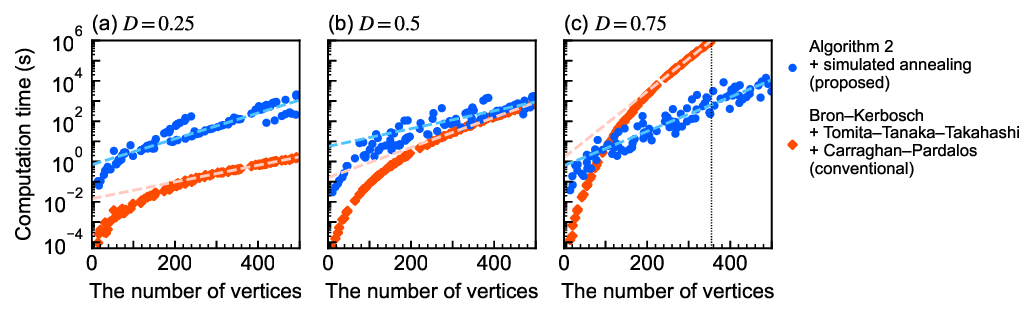}
    \caption{Computation times required to enumerate all maximum cliques in random graphs
             with different graph densities $D$ and different numbers of vertices.
             The blue circles indicate the computation times of Algorithm~2 using SA.
             For each data point, the mean computation time of all successful cases
             out of 100 independent runs was calculated. For graphs with density $D=0.75$
             and more than 355 vertices [the right-hand side of the vertical dotted line
             in panel (c)], where the conventional exact algorithm did not terminate
             even after 10 days and thus the true maximum cliques are unknown, we treat
             runs in which all largest cliques found across the 100 runs were obtained
             as ``successful" cases for the purpose of estimating the computation time.
             The relative standard errors had a mean of 0.01 and a maximum of 0.07;
             thus, the error bars (not shown) are shorter than the diameter of the blue circles.
             The sky-blue dashed lines represent linear fits to the computation times indicated
             by the blue circles for graphs with more than 250 vertices. The red diamonds indicate
             the computation times of the conventional algorithm (Bron--Kerbosch combined
             with the enhancements by Tomita--Tanaka--Takahashi and Carraghan--Pardalos).
             For cases with computation times shorter than seven days, the average of
             10 independent runs was taken. For cases with computation times longer than
             seven days, only one run was conducted to evaluate the computation time.
             The relative standard errors had a mean of 0.01 and a maximum of 0.33;
             thus, the error bars (not shown) are shorter than the size of the red
             diamonds. The light-pink dashed lines represent linear fits to
             the computation times indicated by the red diamonds for graphs
             with more than 250 vertices.}
    \label{fig:computation-time}
\end{figure*}

\begin{table}
    \centering
    \caption{Computation time scaling with respect to the number of vertices $|V_G|$,
             computed through the linear fitting shown in Fig.~\ref{fig:computation-time}.}
    \begin{tabular}{lll}
         \hline
         Density $D$ & Algorithm~2 + SA & Conventional Algorithm\\
         \hline
         0.25 & $O(1.015^{|V_G|})$ & $O(1.010^{|V_G|})$ \\
         0.5  & $O(1.010^{|V_G|})$ & $O(1.017^{|V_G|})$ \\
         0.75 & $O(1.020^{|V_G|})$ & $O(1.038^{|V_G|})$\footnotemark \\
         \hline
    \end{tabular}
    \footnotetext[1]{
        This value was estimated for graphs where the number of vertices
        is less than 355 because the conventional algorithm did not
        terminate even after 10 days.
    }
    \label{tbl:computation-time-scaling}
\end{table}

In terms of the computation time scaling with respect to the number of
vertices, the performance of the conventional algorithm is more susceptible
to the graph density than that of Algorithm~2 using SA. The rate of increase
in the computation time of the conventional algorithm becomes significantly
higher as the graph density increases. In contrast, that of Algorithm~2
only modestly changes with the graph density. Consequently, for sparse graphs
with $D = 0.25$, the conventional algorithm exhibits better performance
than Algorithm~2, while for denser graphs with $D = 0.5$ and $0.75$,
our Algorithm~2 outperforms the conventional algorithm. Note in particular that
for dense and large random graphs with a density of 0.75 and more than 355
vertices, the conventional algorithm did not terminate even after 10 days.

Furthermore, we have recently found that our Algorithm~2 using SA requires
less computation time than the conventional algorithm for maximum clique
problems arising in a chemoinformatics application---atom-to-atom mapping
\cite{Ali2025}. This improvement contributes to achieving accurate and
practical atom-to-atom mapping without relying on any chemical reaction rule
or machine learning.

\subsubsection{Success rate}

Next, we examine the success rate---the fraction of runs in which all maximum cliques
are successfully identified---for Algorithm~2. The statistics of the number of successes
(out of 100 independent runs) for problems with different graph densities
are shown in Fig.~\ref{fig:success-rate}. Here, the success of Algorithm~2 was determined
by whether all maximum cliques obtained by the conventional exact algorithm were enumerated.
The conventional algorithm was able to solve all 100 problems for $D=0.25$ and $0.5$,
but 74 problems for $D=0.75$ within 10 days. Therefore, for $D=0.25$ and $0.5$,
the number of successes over 100 independent runs of Algorithm~2 was calculated
for all 100 problems, whereas for $D=0.75$, it was calculated for the 74 exactly-solved
problems.

\begin{figure*}
    \centering
    \includegraphics[width=\linewidth]{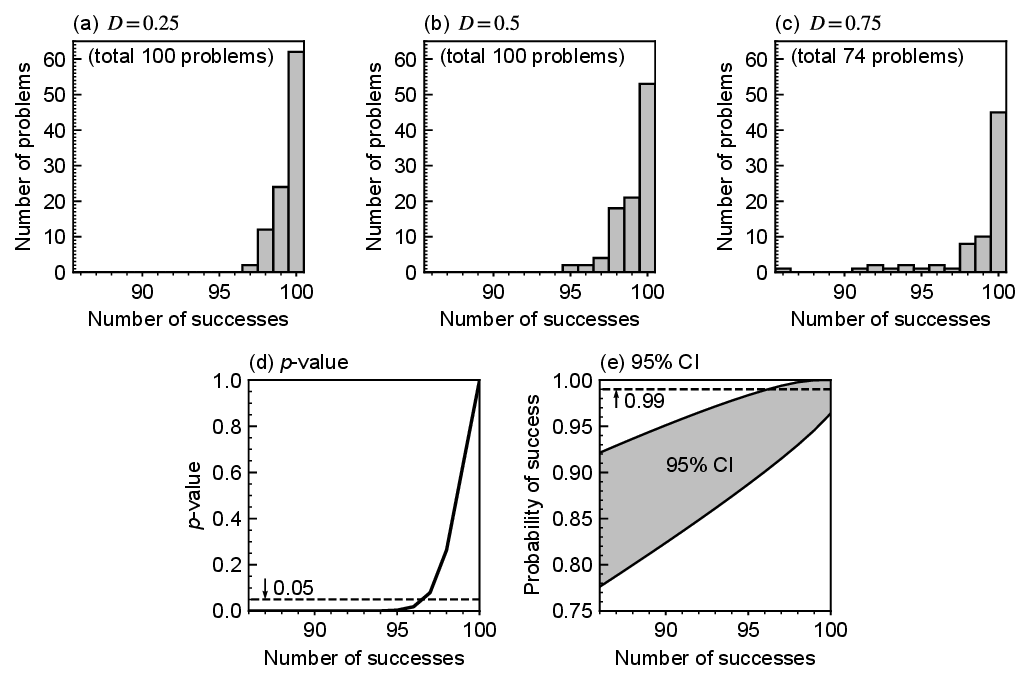}
    \caption{Statistics of the number of successes of Algorithm~2. Panels (a)--(c)
             display histograms of the number of successful runs for problems
             with different graph densities: the height of the bar at $n_\mathrm{success}$
             indicates the number of problems for which the algorithm successfully identified
             all maximum cliques in $n_\mathrm{success}$ out of 100 independent runs.
             Panel (d) shows the $p$-values for each observed number of successes,
             representing the probability of obtaining that number or fewer successes
             under the hypothesis that the true success probability is 0.99.
             The $p$-values were computed from the cumulative distribution function
             of the binomial distribution. Panel (e) presents the 95\% confidence intervals
             (CIs; also referred to as compatibility intervals \cite{Amrhein2019})
             for the estimated success probability based on each observed number of successes.
             The CIs were computed by the method detailed in \cite{Clopper1934,Thulin2014}.}
    \label{fig:success-rate}
\end{figure*}

When the cost-ordered fair sampling condition holds, Algorithm~2 guarantees
a success probability greater than 0.99 for a single run with $\epsilon=0.01$.
Consequently, observing fewer than 97 successes in 100 independent runs is
considered incompatible with the theoretical success probability, according to
the criterion that the $p$-value is less than 0.05 and/or the 95\% CI does not
include 0.99 [see Figs.~\ref{fig:success-rate}(d) and (e)]. The number of
incompatible cases increases with $D$: 0, 4, and 10 cases were observed
for $D=0.25$, $0.5$, and $0.75$, respectively. In all failure runs,
the algorithm still found at least one maximum clique but missed
some others, indicating that ground-state sampling using SA does not
always satisfy the fair sampling condition.

To assess the fairness of the ground-state sampling using SA, we conducted
chi-squared tests for each problem, using samples obtained during
the 100 independent runs of Algorithm~2. The results are shown in
Fig.~\ref{fig:fairness} and Table~\ref{tbl:category-statistics}.

\begin{figure*}
    \centering
    \includegraphics[width=0.66\linewidth]{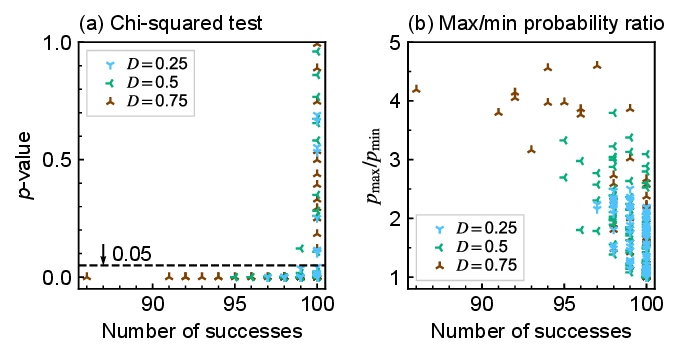}
    \caption{Fairness of the ground-state sampling using SA. Panel (a) shows
             the relationship between the number of successes and the $p$-value of
             the chi-squared test, assessing the fairness of the ground-state sampling.
             A $p$-value closer to 0 suggests that the observed sampling frequencies are
             unlikely under the hypothesis that the ground-state sampling is fair.
             For problems with a unique ground state, the $p$-values are undefined;
             thus, data points for these problems are not plotted in this panel.
             Panel (b) presents the correlation between the number of successes
             and the ratio of the maximum and minimum sampling probabilities
             among ground states, denoted by $p_\mathrm{max}/p_\mathrm{min}$.
             This ratio is also an indicator of the fairness of the ground-state
             sampling: the closer it is to 1, the more fair the ground-state
             sampling is.}
    \label{fig:fairness}
\end{figure*}

\begin{table}
    \centering
    \caption{Statistical categorization of sampling fairness
             for each problem\protect\footnotemark[1].}
    \begin{tabular}{lllll}
    \hline
    Density $D$ & All\footnotemark[2] & Unique\footnotemark[3]
    & \multicolumn{2}{c}{Multiple\footnotemark[4]} \\
    & & & ``Fair'' & ``Unfair'' \\
    \hline
    0.25 & 100 (0) & 16 (0) & 8 (0) & 76 (0) \\
    0.5  & 100 (4) & 17 (0) & 11 (0) & 72 (4) \\
    0.75 & 74 (10) & 11 (0) & 12 (0) & 51 (10) \\
    \hline
    \end{tabular}
    \footnotetext[1]{The numbers in parentheses indicate the number of incompatible cases,
                     i.e., problems for which the observed success count was less than 97.}
    \footnotetext[2]{The total number of problems solved by the conventional exact algorithm.}
    \footnotetext[3]{The number of problems with a unique ground state.}
    \footnotetext[4]{``Fair" (respectively,``Unfair") represents problems
                     with multiple ground states where the ground-state sampling
                     using SA had a $p$-value of the chi-squared test greater than or
                     equal to (respectively, less than) 0.05.}
    \label{tbl:category-statistics}
\end{table}

In Table~\ref{tbl:category-statistics}, we tentatively categorize sampling probability
distributions on multiple ground states into ``fair" and ``unfair" based on
the $p$-values of the chi-squared tests
\footnote{Note that categorizing sampling
probability distributions as ``fair'' and ``unfair'' based on the $p$-values is
not definitive; it only indicates whether each distribution is considered compatible
with the fair sampling condition under a specified criterion.}.
The numbers in parentheses in the table indicate the number of problems
where the number of successes was fewer than 97, suggesting incompatibility
with the theoretical guarantee of Algorithm~2. Note here that all cases
incompatible with the theoretical guarantee are assigned to ``unfair",
as expected. Furthermore, it is worth noting that there are many ``unfair" cases
with an estimated success probability compatible with 0.99. These facts can also
be confirmed by the relationship between the number of successes and the $p$-value
of the chi-squared test presented in Fig.~\ref{fig:fairness}(a).

As another indicator of the fairness of the ground-state sampling, we also
calculated the ratio of the maximum and minimum sampling probabilities
among ground states, denoted by $p_\mathrm{max}/p_\mathrm{min}$.
Figure~\ref{fig:fairness}(b) indicates a moderate negative correlation between
the number of successes and $p_\mathrm{max}/p_\mathrm{min}$,
with a Pearson correlation coefficient value of -0.69. As expected,
a larger variation in the sampling probability tends to
result in fewer successes.

Finally, we calculated the solution coverage defined as the number of
collected solutions divided by the total number of the optimal solutions.
For any problem the algorithm solved, the mean solution coverage of the 100 runs
was greater than or equal to 0.99. This implies that even though the algorithm
fails to enumerate all optimal solutions, only a few solutions are uncollected.
In fact, the number of uncollected solutions in failure runs was typically
one, and two or more uncollected solutions were observed
only in seven problems.

In summary, ground-state sampling using SA is not necessarily fair
in the present setting. However, our Algorithm~2 still works effectively
even for such ``unfair" cases with high success probability and/or high
solution coverage, though there is no theoretical guarantee of the success
probability for such cases yet. To theoretically ensure the success
probability, one can employ other samplers such as the Grover-mixer
quantum alternating operator ansatz algorithm \cite{Bartschi2020},
for which the fair sampling condition is theoretically guaranteed.
Alternatively, it should be helpful to extend the present algorithm
and theory to allow variation in sampling probability up to
a user-specified value of $p_\mathrm{max}/p_\mathrm{min}$.

\section{Conclusions} \label{sec:conclusion}

We have developed enumeration algorithms for combinatorial problems,
specifically (1) constraint satisfaction and (2) combinatorial optimization
problems, using Ising machines as solution samplers. Appropriate stopping
criteria for solution sampling have been derived based on the cost-ordered
fair sampler model. If the solution sampling satisfies the cost-ordered and
fair sampling conditions, the proposed algorithms have theoretical guarantees
that the failure probability is below a user-specified value $\epsilon$.
Various types of physics-based Ising machines are likely to be employed
to implement (approximate) cost-ordered fair samplers. Even though
the sampling process may not strictly satisfy the cost-ordered and
fair sampling conditions in practice, the proposed algorithms can still
function effectively with high success probability and/or high solution
coverage, as demonstrated in maximum clique enumeration using SA.
Furthermore, we showed that Algorithm~2 using SA outperforms
a conventional algorithm for maximum clique enumeration on dense
random graphs and a chemoinformatics application \cite{Ali2025}.

The proposed algorithms rely on the cost-ordered fair sampler model:
more preferred solutions are sampled more frequently, and equally preferred
solutions are sampled with equal probability. Although this model captures
desirable features of samplers for optimization and serves as an archetypal
model for this initial algorithm development, relaxing the cost-ordered
and/or fair sampling conditions should be helpful for expanding
the applicable domain of sampling-based enumeration algorithms.

Moreover, although we have focused on the use of Ising machines in this article,
the proposed algorithms can also be combined with other types of solution samplers
that can be regarded as (approximate) cost-ordered fair samplers. For example,
when combined with a Boltzmann sampler of molecular structures (in an appropriate
discretized representation), our algorithm can determine when to stop
exploring the molecular energy landscape \cite{Wales2003,Wales2018} without
missing the global minimum. Developing sampling-based enumeration algorithms
combined with samplers in various fields is a promising research direction.

We anticipate that the algorithms introduced in this work will drive
future innovations in sampling-based enumeration methods and inspire
a wide range of interdisciplinary applications.

\begin{acknowledgments}
The authors thank Seiji Akiyama for helpful discussions on
maximum clique enumeration and its application to atom-to-atom mapping.
This work was supported by JST, PRESTO Grant Number JPMJPR2018, Japan,
and by the Institute for Chemical Reaction Design and Discovery (ICReDD),
established by the World Premier International Research Initiative (WPI),
MEXT, Japan. This research was also conducted as part of collaboration with
Hitachi Hokkaido University Laboratory, established by Hitachi, Ltd.
at Hokkaido University, and performed partially under the Cooperative Research
Program of ``Network Joint Research Center for Materials and Devices (MEXT).''
\end{acknowledgments}

\section*{Author Contributions}
Y.M. conceptualized the work, developed the algorithms, and
conducted the mathematical analysis and the numerical experiments. 
M.A. contributed to the design of the statistical analysis,
and M.A. and T.K. validated the mathematical proofs.
Y.M. wrote the original draft, and M.A. and T.K.
reviewed and edited the manuscript.

\appendix

\section{Theoretical Analysis} \label{appx:theory}

This appendix presents a theoretical analysis of the failure probabilities
of Algorithms~1 and 2 proposed in this article.

\subsection{Notation} \label{appx:notation}

Let $X$ be a finite set with cardinality $n$, and let $p \colon X \to [0, 1]$
be a discrete probability distribution (probability mass function) on $X$.
Consider the sampling process from $X$, comprising independent trials,
each of which is governed by $p$. We define random variables involved in
the sampling process under the probability distribution $p$ as follows
\footnote{
  We denote the probability measure associated with
  the sampling process simply as $P$. Although this measure depends
  on $X$ and $p$, we do not explicitly indicate this dependence
  in the present article, as the symbol $P$ is always used together
  with random variable symbols indicating the distribution $p$.
}:
\begin{itemize}
\item $x^{(p)}_\tau$: The item sampled at the $\tau$th trial. By definition,
                      $P(x^{(p)}_\tau=x) = p(x)$ for any $x \in X$.
\item $S^{(p)}_\tau$: The set of distinct items that have been sampled by
                      the $\tau$th trial.
\item $\mathfrak{x}^{(p)}_i$: The $i$th distinct item sampled during the process;
                              that is, the $i$th new distinct item not previously
                              sampled.
\item $\mathfrak{S}^{(p)}_i$: The set of the first $i$ distinct sampled items,
                              i.e., $\{\mathfrak{x}_j^{(p)} \mid j=1,\dots,i\}$.
\item $T^{(p)}_m$: The number of trials needed to collect $m$ distinct items;
                   equivalently, the trial number at which the $m$th distinct
                   item $\mathfrak{x}^{(p)}_m$ is first sampled.
\item $t^{(p)}_m$: The number of trials needed to sample the $m$th distinct item
                   after having sampled $m-1$ distinct items, i.e.,
                   $T^{(p)}_m - T^{(p)}_{m-1}$.
\end{itemize}

For instance, in the sample sequence---red (trial 1), yellow (trial 2),
red (trial 3), and blue (trial 4):
\begin{itemize}
\item $x^{(p)}_3 = \text{red}$, $\mathfrak{x}^{(p)}_3 = \text{blue}$.
\item $S^{(p)}_3 = \{\text{red}, \text{yellow}\}$,
      $\mathfrak{S}^{(p)}_3 = \{\text{red}, \text{yellow}, \text{blue}\}$.
\item $T^{(p)}_3 = 4$, $t^{(p)}_3 = 2$.
\end{itemize}

Furthermore, we define $S^{(p)}_0$ and $\mathfrak{S}^{(p)}_0$ as
the empty set and $T^{(p)}_0$ as zero, initializing the process.
In the special case where $p$ is a discrete uniform distribution,
we replace the superscript $(p)$ in the notation by $(n)$
(the cardinality of $X$), e.g., $T^{(n)}_m$.

\subsection{Failure Probability of Algorithm~1} \label{appx:algorithm-1}

In this subsection, we evaluate the failure probability of Algorithm~1.
Since Algorithm~1 utilizes a fair sampler of feasible solutions,
we assume that $X$ is a feasible solution set with cardinality $n$
and the sampling probability distribution $p$ is the discrete uniform
distribution on $X$. Furthermore, Algorithm~1 samples one feasible solution
at the beginning (see line 1 of the pseudocode), and thus Algorithm~1 always
succeeds when $n = 1$. Hence, we assume $n \ge 2$.

Algorithm~1 fails to enumerate all $n$ feasible solutions if and only if
it misses a deadline for collecting $m$ distinct feasible solutions
($m = 2, \dots, n$). In other words,
if $T^{(n)}_m > \lceil m\ln(m\kappa_1/\epsilon) \rceil$ for some $m$,
Algorithm~1 halts before collecting all feasible solutions.
Therefore, the failure probability of Algorithm~1 is bounded above by
the sum of $P(T^{(n)}_m > \lceil m\ln(m\kappa_1/\epsilon) \rceil)$
over $m = 2, \dots, n$. Thus, our primary goal is to evaluate
the tail distribution of $T^{(n)}_m$.

We start by evaluating the simplest case where $m=n$,
which corresponds to the classical coupon collector's problem.

\begin{lemma} \label{lemma:tail-distribution-complete}
Suppose $X$ is a finite set with cardinality $n$ and $p$ is the discrete
uniform distribution on $X$. Let $\epsilon$ be a positive real number
less than one. In the sampling process dictated by $p$, the probability
that $T^{(n)}_n$ exceeds $\lceil n\ln(n/\epsilon) \rceil$ is less than
$\epsilon$:
\begin{equation}
P\left(T^{(n)}_n > \left\lceil n\ln\frac{n}{\epsilon} \right\rceil \right)
< \epsilon.
\end{equation}
\end{lemma}

\begin{proof}
The probability that an element $x \in X$ has not been sampled yet
up to the $\tau$th trial is given by
\begin{equation}
P\left(x \notin S^{(n)}_\tau \right) =
\left(1-\frac{1}{n}\right)^\tau < \mathrm{e}^{-\frac{\tau}{n}}.
\end{equation}
Since $T^{(n)}_n > \tau$ means that there exists $x \in X$ that has
not been sampled yet up to $\tau$, the tail distribution of $T^{(n)}_n$
can be evaluated as follows:
\begin{align}
P\left(T^{(n)}_n > \tau \right)
&= P\left(\bigcup_{x \in X}\{x \notin S^{(n)}_\tau\} \right) \nonumber \\
&\le \sum_{x \in X} P\left(x \notin S^{(n)}_\tau \right) \nonumber \\
&< n\mathrm{e}^{-\frac{\tau}{n}}.
\end{align}
By substituting $\lceil n\ln(n/\epsilon) \rceil$ for $\tau$ in
the above equation, we establish the inequality to be proved.
\end{proof}

Next, we generalize Lemma~\ref{lemma:tail-distribution-complete}
to arbitrary $m\ (\le n)$.

\begin{lemma} \label{lemma:tail-distribution-partial}
Suppose $X$ is a finite set with cardinality $n$ and $p$ is the discrete
uniform distribution on $X$. Let $\epsilon$ be a positive real number
less than one. In the sampling process dictated by $p$, for a positive
integer $m\ (\le n)$, the probability that $T^{(n)}_m$ exceeds
$\lceil m\ln(m/\epsilon) \rceil$ is bounded from above as follows:
\begin{equation}
P\left(T^{(n)}_m > \left\lceil m\ln\frac{m}{\epsilon} \right\rceil \right)
< \left(\frac{m}{n}\right)^{\left\lceil{m\ln\frac{m}{\epsilon}}\right\rceil+1}
  \binom{n}{m} \epsilon.
\end{equation}
\end{lemma}

\begin{proof}
The random variable $t^{(n)}_i$ follows the geometric distribution given by
\begin{equation}
P\left(t^{(n)}_i = \tau_i \right) =
\left(\frac{i-1}{n}\right)^{\tau_i-1} \frac{n-(i-1)}{n}.
\end{equation}
This is because the event that $t^{(n)}_i$ equals $\tau_i$ occurs when
the following two conditions are met. First, during the first $\tau_i-1$
trials, the sampler generates any of the $i-1$ already-sampled items.
Second, at the $\tau_i$th trial, it samples one of the $n-(i-1)$ items
not previously sampled. Furthermore, the random variables
$t^{(n)}_1, t^{(n)}_2, \dots, t^{(n)}_m$ are mutually independent,
as the sampling trials are independent. Consequently, we obtain
\begin{align}
&P\left(t^{(n)}_1 = \tau_1, t^{(n)}_2 = \tau_2,
        \dots, t^{(n)}_m = \tau_m \right) \nonumber \\
&= \prod_{i=1}^{m} \left(\frac{i-1}{n}\right)^{\tau_i-1} \frac{n-(i-1)}{n}
   \nonumber \\
&= \frac{1}{n^{\tau^\prime}}
   \frac{n!}{(n-m)!} \prod_{i=1}^m (i-1)^{\tau_i-1},
\label{eq:joint-probability}
\end{align}
where $\tau^\prime = \sum_{i=1}^m \tau_i$, and the last step follows
the equation $\prod_{i=1}^m [n-(i-1)] = n!/(n-m)!$.

The random variable $T^{(n)}_m$ can be expressed as
$t^{(n)}_1 + t^{(n)}_2 + \cdots + t^{(n)}_m$. If $t^{(n)}_i = \tau_i$
for each $i$ from $1$ to $m$, any combination of positive integers
$\tau_1, \tau_2, \dots, \tau_m$, satisfying the condition
$\tau_1 + \tau_2 + \cdots + \tau_m = \tau^\prime$, results in
$T^{(n)}_m = \tau^\prime$. Let us introduce the set of
such combinations, which is given by
\begin{align}
&\mathcal{C}_m(\tau^\prime) \coloneqq \nonumber \\
&\left\{(\tau_1, \tau_2, \dots, \tau_m) \in \mathbb{N}^m
       \ \middle| \
       \tau_1 + \tau_2 + \cdots + \tau_m = \tau^\prime \right\}.
\end{align}
Now the tail distribution of $T^{(n)}_m$ can be written as
\begin{align}
&P\left(T^{(n)}_m > \tau \right) \nonumber \\
&=\sum_{\tau^\prime = \tau+1}^\infty
  \sum_{\bm{\tau}_{1:m} \in \mathcal{C}_m(\tau^\prime)}
  P\left(t^{(n)}_1 = \tau_1, t^{(n)}_2 = \tau_2,
         \dots, t^{(n)}_m = \tau_m \right) \nonumber \\
&=\sum_{\tau^\prime = \tau+1}^\infty
  \sum_{\bm{\tau}_{1:m} \in \mathcal{C}_m(\tau^\prime)}
  \frac{1}{n^{\tau^\prime}} \frac{n!}{(n-m)!}
  \prod_{i=1}^m (i-1)^{\tau_i-1},
\label{eq:tail-distribution-partial-equality}
\end{align}
where $\tau$ is an arbitrary positive integer, and $\bm{\tau}_{1:m}$
denotes a tuple $(\tau_1, \tau_2, \dots, \tau_m)$ collectively.
The second summation over $\bm{\tau}_{1:m}$ on the right-hand side accounts
for every possible combination of $\tau_1, \tau_2, \dots, \tau_m$ that
satisfies $T^{(n)}_m = \tau^\prime$. The first summation over $\tau^\prime$
covers all cases where $T^{(n)}_m$ exceeds $\tau$.

We further transform the above equation as follows:
\begin{align}
&P\left(T^{(n)}_m > \tau \right) \nonumber \\
&=\sum_{\tau^\prime = \tau+1}^\infty
  \sum_{\bm{\tau}_{1:m} \in \mathcal{C}_m(\tau^\prime)}
  \left(\frac{m}{n}\right)^{\tau^\prime} \frac{n!}{(n-m)!m!}
  \frac{m!}{m^{\tau^\prime}} \prod_{i=1}^m (i-1)^{\tau_i-1} \nonumber \\
&\le \left(\frac{m}{n}\right)^{\tau+1} \binom{n}{m}
  \sum_{\tau^\prime = \tau+1}^\infty
  \sum_{\bm{\tau}_{1:m} \in \mathcal{C}_m(\tau^\prime)}
  \frac{m!}{m^{\tau^\prime}} \prod_{i=1}^m (i-1)^{\tau_i-1} \nonumber \\
&=\left(\frac{m}{n}\right)^{\tau+1} \binom{n}{m}
  P\left(T^{(m)}_m > \tau \right).
\end{align}
In the transformation from the first line to the second line,
we replace the factor $(m/n)^{\tau^\prime}$ with
$(m/n)^{\tau+1}$ because $(m/n) \le 1$ and $\tau^\prime \ge \tau+1$.
The last step of the transformation is according to
Eq.~\eqref{eq:tail-distribution-partial-equality} where $n$ is replaced by $m$.
Substituting $\lceil m\ln(m/\epsilon) \rceil$ for $\tau$ in the above
equation and applying Lemma~\ref{lemma:tail-distribution-complete} complete
the proof.
\end{proof}

This tail distribution estimate may be roughly interpreted as follows:
Consider an event where the sampler generates solutions only from
an $m$-element subset of $X$. The probability that this event occurs
consecutively at least until the $(\tau+1)$th trial is $(m/n)^{\tau+1}$.
Furthermore, under this event, the probability that the number of trials
needed to collect all $m$ solutions in this subset exceeds $\tau$ is
$P(T^{(m)}_m > \tau)$. Considering all possible combinations of
$m$ solutions from $X$, an upper bound for the probability that
$T^{(n)}_m$ exceeds $\tau$ would be given by
\begin{equation}
\binom{n}{m} \left(\frac{m}{n}\right)^{\tau+1}
P\left(T^{(m)}_m > \tau \right).
\end{equation}
This expression appears in the last equation of the above proof.

We now have an upper bound estimate for the tail distribution of $T^{(n)}_m$.
To calculate an upper bound for the failure probability of Algorithm~1, we will
sum $P(T^{(n)}_m > \lceil m\ln(m\kappa_1/\epsilon) \rceil)$ over $m=2$ to $n$.
However, the upper bound for the tail distribution derived in
Lemma~\ref{lemma:tail-distribution-partial} is still complex and difficult
to sum over $m$. Therefore, our next goal is to simplify the right-hand side
of the inequality in Lemma~\ref{lemma:tail-distribution-partial}.

\begin{lemma} \label{lemma:prefactor-upper-bound}
Let $n$ and $m$ be positive integers satisfying $2 \le m \le n$, and let
$\epsilon$ be a positive real number less than one. Then, the following
inequality holds:
\begin{equation}
\left(\frac{m}{n}\right)^{\left\lceil{m\ln\frac{m}{\epsilon}}\right\rceil}
\binom{n}{m}
< \left(\frac{m}{n}\right)^{\alpha m},
\end{equation}
where $\alpha \coloneqq \ln(1/\epsilon) - 1$.
\end{lemma}

\begin{proof}
Define a function $g$ by the following expression:
\begin{equation}
g(u) \coloneqq
\left(\frac{m}{u}\right)^{\left\lceil{m\ln\frac{m}{\epsilon}}\right\rceil}
\frac{\prod_{i=1}^m [u-(i-1)]}{m!} \quad (u \ge m).
\end{equation}
The inequality to be proven can be expressed as $g(n) \le (m/n)^{\alpha m}$.
Differentiating $\ln g(u)$ with respect to $u$ gives
\begin{align}
\frac{\mathrm{d}}{\mathrm{d}u} \ln g(u)
&= -\frac{\left\lceil{m\ln\frac{m}{\epsilon}}\right\rceil}{u} +
    \sum_{i=1}^m \frac{1}{u-(i-1)} \nonumber \\
&= \frac{1}{u}
   \left[\sum_{i=1}^m \frac{u}{u-(i-1)} -
         \left\lceil{m\ln\frac{m}{\epsilon}}\right\rceil \right].
\end{align}
The summation $\sum_{i=1}^m u/[u-(i-1)]$ decreases monotonically
as $u$ increases. Hence, for $u \ge m$, the summation is upper
bounded by $\sum_{i=1}^m m/[m-(i-1)]\ [=m(1+1/2+\dots+1/m)]$,
which is the $m$th harmonic number multiplied by $m$. Furthermore,
the $m$th harmonic number ($m \ge 2$) can be evaluated as
\begin{align}
\sum_{i=1}^m \frac{1}{m-(i-1)}
&= 1 + \sum_{k=2}^m \frac{1}{k} \nonumber \\
&< 1 + \int_1^m \frac{\mathrm{d}s}{s} \nonumber \\
&= 1 + \ln m.
\end{align}
Therefore, we obtain
\begin{align}
\frac{\mathrm{d}}{\mathrm{d}u} \ln g(u)
&< \frac{1}{u}
  \left[m(1+\ln m) - m\ln\frac{m}{\epsilon}\right] \nonumber \\
&= -\frac{\alpha m}{u},
\end{align}
where $\alpha = \ln(1/\epsilon) - 1$.
Integrating both sides of this inequality from $m$ to $n$ yields
\begin{equation}
\ln \frac{g(n)}{g(m)} < -\alpha m \ln\frac{n}{m}.
\end{equation}
Because $g(m) = 1$, this inequality implies
\begin{equation}
g(n) =
\left(\frac{m}{n}\right)^{\left\lceil{m\ln\frac{m}{\epsilon}}\right\rceil}
\binom{n}{m}
< \left(\frac{m}{n}\right)^{\alpha m}.
\end{equation}
This concludes the proof.
\end{proof}

We further simplify the upper bound as follows:

\begin{lemma} \label{lemma:prefactor-upper-bound-2}
Let $n$ and $m$ be positive integers satisfying $m \le n$, and let
$\alpha$ be a positive real number. Then the following inequalities hold:
\begin{align}
\left(\frac{m}{n}\right)^{\alpha m}
&\le \left(\frac{2}{n}\right)^{2\alpha} \mathrm{e}^{-\beta(m-2)},
& &\text{if}\quad 2 \le m < \frac{n}{\mathrm{e}}, \\
\left(\frac{m}{n}\right)^{\alpha m}
&\le \mathrm{e}^{\frac{\alpha}{\mathrm{e}-1}(m-n)},
& &\text{if}\quad \frac{n}{\mathrm{e}} < m \le n,
\end{align}
where $\beta$ is defined as
\begin{equation}
\beta \coloneqq
\frac{\frac{1}{\mathrm{e}}+\frac{1}{3}\ln\frac{1}{3}}
     {\frac{1}{\mathrm{e}}-\frac{1}{3}} \alpha.
\label{eq:beta-def}
\end{equation}
\end{lemma}

\begin{proof}
The left-hand side of the inequalities can be written as
\begin{equation}
\left(\frac{m}{n}\right)^{\alpha m} =
\exp\left(n\alpha \left(\frac{m}{n}\right)\ln\left(\frac{m}{n}\right)\right).
\label{eq:exp-ulnu}
\end{equation}
To evaluate the exponent in the above equation, we examine
a function $h$ defined as
\begin{equation}
h(u) \coloneqq u\ln u
\end{equation}
for $u \in [2/n, 1]$. Here, the range of $u$ corresponds
to $2 \le m \le n$ via the relation $u = m/n$.
The graph of $v= u\ln u$ is shown in Fig.~\ref{fig:ulnu}.

\begin{figure}
    \centering
    \includegraphics[width=\linewidth]{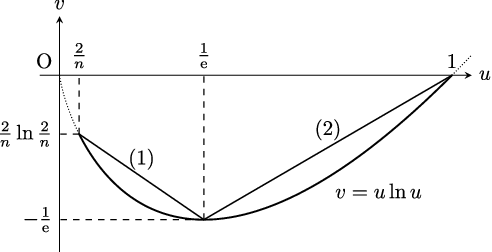}
    \caption{The graph of $v = u\ln u$. The straight lines (1) and (2)
             provide upper bounds for the value of $u \ln u$ in the intervals
             $[2/n, 1/\mathrm{e}]$ and $[1/\mathrm{e}, 1]$, respectively.
             These lines are used to prove the two inequalities
             in Lemma~\ref{lemma:prefactor-upper-bound-2}.}
    \label{fig:ulnu}
\end{figure}

The function $h$ is convex. Therefore, for all $u_1, u_2 \in [2/n, 1]$
and all $\lambda \in [0, 1]$,
\begin{equation}
h((1-\lambda)u_1 + \lambda u_2) \le (1-\lambda)h(u_1) + \lambda h(u_2).
\end{equation}
This inequality is equivalent to
\begin{equation}
h(u) \le h(u_1) + \frac{h(u_2) - h(u_1)}{u_2 - u_1}(u - u_1),
\label{eq:convex-inequality}
\end{equation}
where $u\ [= (1-\lambda)u_1 + \lambda u_2]$ lies between $u_1$ and $u_2$.
The right-hand side of Eq.~\eqref{eq:convex-inequality} represents the line passing
through points $(u_1, h(u_1))$ and $(u_2, h(u_2))$. We apply this inequality
to two intervals: $[2/n, 1/\mathrm{e}]$ and $[1/\mathrm{e}, 1]$,
as shown in Fig.~\ref{fig:ulnu}.

First, suppose $2 \le m < n/\mathrm{e}$. This condition corresponds to
the interval $[2/n, 1/\mathrm{e}]$, implying $n > 2\mathrm{e}$.
Let $u_1 = 2/n$ and $u_2 = 1/\mathrm{e}$. Then Eq.~\eqref{eq:convex-inequality} gives
\begin{equation}
h(u) \le
\frac{2}{n}\ln\frac{2}{n} -
\frac{\frac{1}{\mathrm{e}}+\frac{2}{n}\ln\frac{2}{n}}
     {\frac{1}{\mathrm{e}}-\frac{2}{n}} \left(u-\frac{2}{n}\right)
\label{eq:inequality-ulnu-1}
\end{equation}
for $u \in [2/n, 1/\mathrm{e}]$ [see the line (1) in Fig.~\ref{fig:ulnu}].
We can verify that the coefficient of $u$ on the right-hand side decreases
as $n$ increases. As we can see from Fig.~\ref{fig:ulnu}, when $n$ becomes larger
(i.e., $2/n$ becomes smaller), the slope of the line (1) becomes steeper
in the negative direction. Thus, the coefficient attains its maximum value
when $n=6$, which is the smallest integer satisfying $n > 2\mathrm{e}$:
\begin{equation}
-\frac{\frac{1}{\mathrm{e}}+\frac{2}{n}\ln\frac{2}{n}}
      {\frac{1}{\mathrm{e}}-\frac{2}{n}}
\le - \frac{\frac{1}{\mathrm{e}}+\frac{1}{3}\ln\frac{1}{3}}
           {\frac{1}{\mathrm{e}}-\frac{1}{3}}
= -\frac{\beta}{\alpha}.
\end{equation}
Therefore, we obtain
\begin{align}
\left(\frac{m}{n}\right)^{\alpha m}
&= \exp\left[n\alpha \cdot h\left(\frac{m}{n}\right)\right] \nonumber \\
&\le \exp\left[n\alpha \left(\frac{2}{n}\ln\frac{2}{n}-
     \frac{\beta}{\alpha}\left(\frac{m}{n}-\frac{2}{n}\right)\right)\right] \nonumber \\
&= \left(\frac{2}{n}\right)^{2\alpha} \mathrm{e}^{-\beta(m-2)}
\end{align}
for $2 \le m < n/\mathrm{e}$. This is the first inequality of the lemma.

Next, suppose $n/\mathrm{e} < m \le n$, which corresponds to the interval
$[1/\mathrm{e}, 1]$. Let $u_1 = 1$ and $u_2 = 1/\mathrm{e}$.
Then Eq.~\eqref{eq:convex-inequality} gives
\begin{equation}
h(u) \le 0 + \frac{-\frac{1}{\mathrm{e}}-0}{\frac{1}{\mathrm{e}}-1}(u-1)
= \frac{u-1}{\mathrm{e}-1}
\end{equation}
for $u \in [1/\mathrm{e}, 1]$ [see the line (2) in Fig.~\ref{fig:ulnu}].
Therefore, applying this inequality with $u = m/n$ to Eq.~\eqref{eq:exp-ulnu},
we obtain
\begin{align}
\left(\frac{m}{n}\right)^{\alpha m}
&\le \exp\left(n\alpha \frac{\frac{m}{n}-1}{\mathrm{e}-1}\right) \nonumber\\
&= \exp\left(\frac{\alpha}{\mathrm{e}-1}(m-n)\right)
\end{align}
for $n/\mathrm{e} < m \le n$. This is the second inequality of the lemma.
\end{proof}

Now we are ready to prove that the failure probability of Algorithm~1
is less than $\epsilon$.

\begin{theorem} \label{theorem:algorithm-1}
Let $X$ be the set of all feasible solutions to be enumerated.
Suppose that the number of feasible solutions, denoted by $n$, is unknown.
Let $\epsilon \in (0, 1/\mathrm{e})$ be a user-specified tolerance for
the failure probability of the exhaustive solution enumeration. Then,
using a fair sampler that follows the discrete uniform distribution on $X$,
Algorithm~1 successfully enumerates all feasible solutions in $X$ with
a probability greater than $1-\epsilon$, regardless of the unknown value of $n$.
\end{theorem}

\begin{proof}
Since Algorithm~1 always succeeds when $n=1$, it is sufficient to prove
the theorem for $n \ge 2$. Algorithm~1 fails to exhaustively enumerate
all solutions if and only if the number of samples needed to collect
$m$ distinct solutions, denoted by $T^{(n)}_m$, exceeds the deadline
$\lceil m\ln(m\kappa_1/\epsilon) \rceil$ for some positive integer
$m \in [2, n]$. Here, $\kappa_1$ is defined in Eq.~\eqref{eq:kappa1-def}.
Note that $\kappa_1 > 1$, and thus $(\epsilon/\kappa_1) < 1$.
Therefore, the failure probability of Algorithm~1 can be evaluated as
\begin{align}
&P\left(\bigcup_{m=2}^n \left\{
        T^{(n)}_m > \left\lceil m\ln\frac{m\kappa_1}{\epsilon} \right\rceil
       \right\}\right) \nonumber\\
&\le \sum_{m=2}^n P\left(
        T^{(n)}_m > \left\lceil m\ln\frac{m\kappa_1}{\epsilon} \right\rceil
      \right) \nonumber\\
&< \sum_{m=2}^n
   \left(\frac{m}{n}\right)^{\left(\ln\frac{\kappa_1}{\epsilon}-1\right) m}
   \left(\frac{\epsilon}{\kappa_1}\right) \nonumber\\
&< \sum_{m=2}^n
   \left(\frac{m}{n}\right)^{\alpha m}
   \left(\frac{\epsilon}{\kappa_1}\right)
\end{align}
The second last step in the derivation follows from
Lemmas~\ref{lemma:tail-distribution-partial} and \ref{lemma:prefactor-upper-bound}
[with $\epsilon$ in these lemmas replaced by $\epsilon/\kappa_1\ (<1)$],
and the inequality $m/n \le 1$. In the final step, $\ln(\kappa_1/\epsilon)-1$
is replaced by $\alpha\ [\coloneqq \ln(1/\epsilon)-1]$ because
$\ln(\kappa_1/\epsilon) > \ln(1/\epsilon)$.

Now we aim to demonstrate that the summation $\sum_{m=2}^n (m/n)^{\alpha m}$
is less than $\kappa_1$, which makes the right-hand side of the above equation
bounded above by $\epsilon$. Using the inequalities given in
Lemma~\ref{lemma:prefactor-upper-bound-2}, we get
\begin{equation}
\sum_{m=2}^n \left(\frac{m}{n}\right)^{\alpha m}
\le \sum_{m=2}^{\left\lfloor \frac{n}{\mathrm{e}} \right\rfloor}
    \left(\frac{2}{n}\right)^{2\alpha} \mathrm{e}^{-\beta(m-2)} +
    \sum_{m = \left\lceil \frac{n}{\mathrm{e}} \right\rceil}^n
    \mathrm{e}^{\frac{\alpha}{\mathrm{e}-1}(m-n)},
\label{eq:split-summation}
\end{equation}
where $\lfloor\ \rfloor$ denotes the floor function.
For the right-hand side, the first summation is considered to be zero
when $\left\lfloor \frac{n}{\mathrm{e}} \right\rfloor < 2$.
Additionally, when $\lceil n/\mathrm{e} \rceil = 1$, it is considered that
the variable $m$ in the second summation starts at $2$ instead of $1$.
Since the first term contributes only if $n > 2\mathrm{e} = 5.43\cdots$,
the factor $(2/n)^{2\alpha}$ in the first summation can be bounded above
by $(2/6)^{2\alpha}= 3^{-2\alpha}$. Furthermore, we can bound the finite
summations from above by their corresponding infinite geometric series.
Therefore, we obtain
\begin{equation}
\sum_{m=2}^n \left(\frac{m}{n}\right)^{\alpha m}
< 3^{-2\alpha}
  \sum_{m=2}^\infty \mathrm{e}^{-\beta(m-2)} +
  \sum_{m^\prime = 0}^\infty
  \mathrm{e}^{-\frac{\alpha}{\mathrm{e}-1}m^\prime},
\end{equation}
where $m^\prime$ denotes $n-m$. Since $\epsilon$ is set to be
less than $1/\mathrm{e}$, the parameter $\alpha\ [= \ln(1/\epsilon)-1]$
is positive, which also implies that the parameter $\beta$, given in
Eq.~\eqref{eq:beta-def}, is positive. Thus, the common ratios of the geometric series,
$\exp(-\beta)$ and $\exp(-\alpha/(\mathrm{e}-1))$, are less than one.
Consequently, these geometric series converge, which leads to
\begin{equation}
\sum_{m=2}^n \left(\frac{m}{n}\right)^{\alpha m}
< \frac{3^{-2\alpha}}{1-\mathrm{e}^{-\beta}} +
  \frac{1}{1-\mathrm{e}^{-\frac{\alpha}{\mathrm{e}-1}}}.
\end{equation}
The right-hand side equals $\kappa_1$ by definition. Therefore, we conclude
that the failure probability of Algorithm~1 remains strictly below $\epsilon$,
irrespective of the value of $n$.
\end{proof}

This proof clarifies that $\kappa_1$ is designed to bound the sum of
the failure probability at each deadline for $m \in [2, n]$.
In other words, $\kappa_1$ compensates for the increased error chances
caused by checking the number of collected solutions at every deadline
for $m \in [2, n]$---which is necessitated by the lack of information
about $n$.

To derive $\kappa_1$, we replaced the finite summations by the infinite
summations. This transformation effectively removes the dependence on
the unknown value of $n$. Although infinitely many redundant terms are
included in the infinite summations, they become exponentially small as
the index increases; thus, convergence is expected to be fast.
Indeed, the value of $\kappa_1$ is around 1.14 when $\epsilon=0.01$,
which is only slightly larger than its lower bound 1.

\subsection{Failure Probability of Algorithm~2} \label{appx:algorithm-2}

In this subsection, we evaluate the failure probability of Algorithm~2.
Since Algorithm~2 utilizes a cost-ordered fair sampler of feasible solutions,
we assume that $X$ is a feasible solution set with cardinality $n$ and
$p$ is a probability distribution on $X$ satisfying the fair and cost-ordered
sampling conditions given in Eq.~\eqref{eq:cost-ordered-fair}. Furthermore,
Algorithm~2 initially samples one feasible solution (see line 22 in the pseudocode),
and thus it always succeeds when $n=1$. Hence, we assume $n \ge 2$.

When the current minimum cost among sampled solutions is $\theta$,
the algorithm discards any sample with cost exceeding $\theta$.
In other words, the sampler effectively generates samples from the set of
feasible solutions with cost lower than or equal to $\theta$.
To analyze Algorithm~2 equipped with this mechanism, we introduce
the following notation: let us define $X_\theta$ and $Y_\theta$ as
\begin{align}
X_\theta &\coloneqq \{x \in X \mid f(x) \le \theta \}, \\
Y_\theta &\coloneqq \{x \in X \mid f(x) = \theta \}.
\end{align}
We denote the cardinalities of $X_\theta$ and $Y_\theta$ by
$n_\theta$ and $l_\theta$, respectively. The sampling probability
distribution for $\theta$, denoted by $p_\theta$, is defined as:
\begin{equation}
p_\theta(x) \coloneqq
\begin{cases}
\frac{p(x)}{\sum_{x^\prime \in X_\theta}p(x^\prime)},
& \text{if}\ x \in X_\theta, \\
0, & \text{if}\ x \notin X_\theta.
\end{cases}
\end{equation}
The second line corresponds to rejecting any $x \notin X_\theta$.
This sampling distribution also satisfies the cost-ordered and fair
sampling conditions: for any two feasible solutions $x_1, x_2 \in X_\theta$,
\begin{align}
f(x_1) < f(x_2) \Rightarrow p_\theta(x_1) \ge p_\theta(x_2), \\
f(x_1) = f(x_2) \Rightarrow p_\theta(x_1) = p_\theta(x_2).
\end{align}
Additionally, for $\theta > \min_{x \in X} f(x)$, the cost value for
$x \in X_\theta \setminus Y_\theta$ is less than that for any $y \in Y_\theta$
by definition, and thus
\begin{equation}
y \in Y_\theta \ \text{and}\ x \in X_\theta \setminus Y_\theta
\Rightarrow p_\theta(y) \le p_\theta(x).
\end{equation}
This condition is used in Lemma~\ref{lemma:tail-distribution-subset},
as described below. We emphasize here that although $X_\theta$, $Y_\theta$,
$p_\theta$, and related quantities are not directly accessible in practice
without complete knowledge of the energy (or cost) landscape and
the sampling distribution, introducing these notations enables
the subsequent general analysis and allows us to derive an upper bound
on the failure probability of Algorithm~2 that does not depend on
these unknown quantities.

We first analyze a failure scenario where Algorithm~2 terminates before
sampling any optimal solution. In this scenario, the algorithm returns $m^\prime$
feasible solutions whose cost value is $\theta$, where $1 \le m^\prime \le l_\theta$
and $\theta > \min_{x \in X} f(x)$. Such a failure occurs if the following
conditions are met during the sampling process governed by $p_\theta$ over $X_\theta$:
\begin{itemize}
\item The first $m^\prime$ sampled distinct solutions all have cost value $\theta$;
      that is, $\mathfrak{S}^{(p_\theta)}_{m^\prime} \subset Y_\theta$.
\item The time to obtain the $(m^\prime+1)$th distinct solution,
      $T^{(p_\theta)}_{m^\prime + 1}$, exceeds the deadline for
      collecting $m^\prime + 1$ distinct solutions.
\end{itemize}
Note that $X_\theta$ contains at least $m^\prime + 1$ distinct solutions
because $n_\theta \ge l_\theta + 1\ [\because \theta > \min_{x \in X} f(x)]$,
and consequently $T^{(p_\theta)}_{m^\prime + 1}$ is well-defined.

The following lemma provides an upper bound for the probability of
such an event. [For notational simplicity, we omit subscript $\theta$
and replace $m^\prime$ with $m - 1$ ($2 \le m \le l_\theta + 1$) in the lemma.]

\begin{lemma} \label{lemma:tail-distribution-subset}
Let $X$ be a finite set with cardinality $n$, and let $Y$ be a proper subset
of $X$ with cardinality $l$. Assume that the probability distribution $p$
governing the sampling process from $X$ satisfies the conditions:
(1) $y_1 \in Y \ \text{and}\ y_2 \in Y \Rightarrow p(y_1) = p(y_2)$;
(2) $y \in Y \ \text{and}\ x \in X \setminus Y \Rightarrow p(y) \le p(x)$.
Then, for any positive integer $m \in [2, l+1]$ and any positive real
number $\epsilon$ less than one, the probability that $T^{(p)}_m$ exceeds
$\lceil m\ln(m/\epsilon) \rceil$ and $\mathfrak{S}^{(p)}_{m-1}$ is a subset
of $Y$ is bounded from above as follows:
\begin{align}
&P\left(T^{(p)}_m > \left\lceil m\ln\frac{m}{\epsilon} \right\rceil,\
       \mathfrak{S}^{(p)}_{m-1} \subset Y \right) \nonumber\\
&< \left(\frac{m}{n}\right)^{\left\lceil{m\ln\frac{m}{\epsilon}}\right\rceil+1}
   \binom{n}{m} \epsilon.
\end{align}
\end{lemma}

\begin{proof}
Due to the first condition on $p$, we can denote the equal probability
of sampling $y \in Y$ as $p_Y$, i.e., $p(y) = p_Y$ for all $y \in Y$.
This sampling probability satisfies $p_Y \le 1/n$, because if $p_Y > 1/n$,
it would violate the unit-measure axiom of probability:
\begin{align}
\sum_{x\in X}p(x) &= \sum_{y \in Y}p_Y + \sum_{x\in X \setminus Y}p(x) \nonumber\\
&\ge \sum_{y \in Y}p_Y + \sum_{x\in X \setminus Y}p_Y \nonumber \\
&\qquad (\because \text{the second condition on }p) \nonumber \\
&= np_Y > 1.
\end{align}

Given that $\mathfrak{S}^{(p)}_{i-1} \subset Y$, there are $l-(i-1)$
uncollected items in $Y$. The probability of sampling the $i$th new
distinct item $\mathfrak{x}^{(p)}_i$ from these uncollected items in $Y$
at $t^{(p)}_i = \tau_i$ is then calculated as
\begin{align}
&P\left(t^{(p)}_i = \tau_i,\ \mathfrak{x}^{(p)}_i \in Y
       \ \middle| \ \mathfrak{S}^{(p)}_{i-1} \subset Y \right) \nonumber \\
&= [(i-1)p_Y]^{\tau_i-1} \left[l-(i-1)\right]p_Y.
\end{align}
Similarly, the probability that $t^{(p)}_i$ equals $\tau_i$ is given by
\begin{equation}
P\left(t^{(p)}_i = \tau_i \ \middle| \
       \mathfrak{S}^{(p)}_{i-1} \subset Y \right) =
[(i-1)p_Y]^{\tau_i-1} \left[1-(i-1)p_Y\right],
\end{equation}
because any of the uncollected items in $X$ can be $\mathfrak{x}^{(p)}_i$,
and the probability of sampling such an item is
$1-\sum_{y \in \mathfrak{S}^{(p)}_{i-1}}p_Y = 1 - (i-1)p_Y$.
(Note that for the fair sampling case where $p_Y = 1/n$, this probability
 distribution is equivalent to the geometric distribution, as shown in
 the first equation of the proof of Lemma~\ref{lemma:tail-distribution-partial}.)
Furthermore, the random variables $t^{(p)}_1, t^{(p)}_2, \dots, t^{(p)}_m$
are mutually independent, reflecting the independence of sampling trials.
Additionally, we note that
\begin{equation}
\mathfrak{S}^{(p)}_j \subset Y
\iff \bigwedge_{i=1}^j \mathfrak{x}^{(p)}_i \in Y.
\end{equation}
Therefore, we get the following equation using the chain rule:
\begin{align}
&P\left(t^{(p)}_1 = \tau_1, t^{(p)}_2 = \tau_2, \dots, t^{(p)}_m = \tau_m,\
        \mathfrak{S}^{(p)}_{m-1} \subset Y \right) \nonumber\\
&= \left[\prod_{i=1}^{m-1}
         P\left(t^{(p)}_i = \tau_i,\ \mathfrak{x}^{(p)}_i \in Y \ \middle| \
                \mathfrak{S}^{(p)}_{i-1} \subset Y \right)\right] \nonumber\\
&\qquad \times P\left(t^{(p)}_m = \tau_m \ \middle| \
                 \mathfrak{S}^{(p)}_{m-1} \subset Y \right) \nonumber\\
&= \left[\prod_{i=1}^{m-1} [(i-1)p_Y]^{\tau_i-1}
         \left[l-(i-1)\right] p_Y\right] \nonumber\\
&\qquad  \times [(m-1)p_Y]^{\tau_m-1} \left[1-(m-1)p_Y\right] \nonumber\\
&= p_Y^{\tau^\prime-1} \left[1-(m-1)p_Y\right]
   \left[\prod_{i=1}^{m-1}\left[l-(i-1)\right]\right] \nonumber\\
&\qquad \times \left[\prod_{i=1}^m(i-1)^{\tau_i-1}\right],
\end{align}
where $\tau^\prime$ denotes $\sum_{i=1}^m \tau_i$. Since $p_Y \le 1/n$
and $1-(m-1)p_Y \le 1$, we can derive the following inequality:
\begin{align}
&p_Y^{\tau^\prime-1}\left[1-(m-1)p_Y\right]
 \prod_{i=1}^{m-1}\left[l-(i-1)\right] \nonumber\\
&\le \frac{1}{n^{\tau^\prime-1}} \prod_{i=1}^{m-1}\left[l-(i-1)\right] \nonumber\\
&= \frac{\prod_{i=1}^m\left[n-(i-1)\right]}{n^{\tau^\prime}}
   \frac{n \prod_{i=1}^{m-1}\left[l-(i-1)\right]}
        {\prod_{i=1}^m\left[n-(i-1)\right]} \nonumber\\
&= \frac{\prod_{i=1}^m\left[n-(i-1)\right]}{n^{\tau^\prime}}
   \prod_{i=1}^{m-1}\frac{(l+1)-i}{n-i} \nonumber\\
&\le \frac{\prod_{i=1}^m\left[n-(i-1)\right]}{n^{\tau^\prime}}.
\end{align}
Here, we derive the final expression following the fact that $Y$ is
a proper subset of $X$, which implies $n \ge l+1$. In summary, we obtain
the inequality
\begin{align}
&P\left(t^{(p)}_1 = \tau_1, t^{(p)}_2 = \tau_2, \dots, t^{(p)}_m = \tau_m,\
        \mathfrak{S}^{(p)}_{m-1} \subset Y \right) \nonumber\\
&\le n^{-\tau^\prime} \prod_{i=1}^m (i-1)^{\tau_i-1}\left[n-(i-1)\right].
\end{align}
According to Eq.~\eqref{eq:joint-probability} in the proof of
Lemma~\ref{lemma:tail-distribution-partial}, the right-hand side equals the joint
probability of $t^{(p)}_1, t^{(p)}_2, \dots, t^{(p)}_m$ for the case
where $p$ is the discrete uniform distribution on $X$, i.e.,
$P\left(t^{(n)}_1=\tau_1, t^{(n)}_2=\tau_2, \dots, t^{(n)}_m=\tau_m \right)$.

As in the proof of Lemma~\ref{lemma:tail-distribution-partial}, we sum the joint
probabilities over all combinations of $\tau_1, \tau_2, \dots, \tau_m$ that
result in $\tau^\prime > \lceil m\ln(m/\epsilon) \rceil$. This calculation
yields
\begin{align}
&P\left(T^{(p)}_m > \left\lceil m\ln\frac{m}{\epsilon} \right\rceil,\
       \mathfrak{S}^{(p)}_{m-1} \subset Y \right) \nonumber\\
&\le P\left(T^{(n)}_m > \left\lceil m\ln\frac{m}{\epsilon} \right\rceil \right).
\end{align}
By applying the inequality established in Lemma~\ref{lemma:tail-distribution-partial}
to the above inequality, we obtain the inequality stated in
the current lemma.
\end{proof}

The inequality in Lemma~\ref{lemma:tail-distribution-subset} can also be roughly
interpreted as follows. The probability that all of the $m-1$ items already sampled
belong to $Y$ is maximized when $p_Y$ is maximized. Under the constraint that
$p(x) \ge p_Y = 1/n$ for all $x \in X$ and $\sum_{x \in X}p(x) = 1$,
this maximization of $p_Y$ results in $p$ being the uniform distribution
on $X$. Thus, we consider the fair sampling case, replacing
superscripts $(p)$ with $(n)$. Obviously,
\begin{align}
&P\left(T^{(n)}_m > \left\lceil m\ln\frac{m}{\epsilon} \right\rceil,\
        \mathfrak{S}^{(n)}_{m-1} \subset Y \right) \nonumber \\
&\le P\left(T^{(n)}_m > \left\lceil m\ln\frac{m}{\epsilon} \right\rceil \right).
\end{align}
This equation is the same as the last equation in the above proof,
except for the difference between the superscripts $(n)$ and $(p)$
on the left-hand sides.

Finally, we prove that the failure probability of Algorithm~2 is
less than $\epsilon$. The following theorem is the main theoretical
result of this article.

\begin{theorem} \label{theorem:algorithm-2}
Let $X$ be the set of all feasible solutions, and let $f: X \to \mathbb{R}$
be the cost function of a combinatorial optimization problem. In addition,
let $\epsilon \in (0, 1/\mathrm{e}^{1.5})$ be a user-specified tolerance
for the failure probability of enumerating all optimal solutions.
Then, using a cost-ordered fair sampler on $X$, Algorithm~2 successfully
enumerates all optimal solutions in $\operatorname{argmin}_{x \in X}f(x)$
with a probability exceeding $1-\epsilon$, regardless of the unknown minimum
value of $f$ and the unknown number of the optimal solutions.
\end{theorem}

\begin{proof}
The failure scenarios of Algorithm~2 fall into two categories:
\begin{enumerate}
\item  Algorithm~2 terminates without having sampled any optimal solution.
\item  Algorithm~2 terminates and returns a proper subset of optimal solutions
       with some optimal solution(s) missing.
\end{enumerate}

First, we analyze failures of the first type: Algorithm~2
samples a feasible solution with cost $\theta > \min_{x \in X}f(x)$
and stops during the sampling process for $\theta$, which is governed by
$p_\theta$. Let $\mathcal{E}_\theta$ denote the event where Algorithm~2
samples a feasible solution with cost $\theta$. Given that
the event $\mathcal{E}_\theta$ occurs, the algorithm stops during
the sampling for $\theta$ when all the first $m-1$ sampled distinct
solutions have cost value $\theta$
(i.e., $\mathfrak{S}^{(p_\theta)}_{m-1} \subset Y_\theta$),
and $T^{(p_\theta)}_m$ exceeds the deadline for collecting
$m$ distinct solutions ($m \in [2, l_\theta+1]$). Then, based on
Lemma~\ref{lemma:tail-distribution-subset} with $\epsilon$ replaced
by $\epsilon/\kappa_2\ (<1)$, we can evaluate the probability
of this failure case, denoted by $P^\mathrm{fail}_\theta$, as follows:
\begin{align}
&P^\mathrm{fail}_\theta \nonumber\\
&= P\left(\bigcup_{m=2}^{l_\theta+1}\left\{
            T^{(p_\theta)}_m >
            \left\lceil m\ln\frac{m\kappa_2}{\epsilon} \right\rceil
            \land
            \mathfrak{S}^{(p_\theta)}_{m-1} \subset Y_\theta
          \right\} \cap \mathcal{E}_\theta \right) \nonumber\\
&\le \sum_{m=2}^{l_\theta+1}
     P\left(T^{(p_\theta)}_m >
            \left\lceil m\ln\frac{m\kappa_2}{\epsilon} \right\rceil,\
            \mathfrak{S}^{(p_\theta)}_{m-1} \subset Y_\theta \right) \nonumber\\
&< \frac{\epsilon}{\kappa_2} \sum_{m=2}^{l_\theta+1}
   \left(\frac{m}{n_\theta}\right)^{
    \left\lceil m\ln\frac{m\kappa_2}{\epsilon} \right\rceil + 1
   } \binom{n_\theta}{m}.
\label{eq:failure-probability-theta}
\end{align}
Using Lemma~\ref{lemma:prefactor-upper-bound} with $\epsilon$ replaced by $\epsilon/\kappa_2$,
each term in the summation can be simplified as
\begin{equation}
P^\mathrm{fail}_\theta
< \frac{\epsilon}{\kappa_2} \sum_{m=2}^{l_\theta+1}
  \left(\frac{m}{n_\theta}\right)^{
    \left(\ln\frac{\kappa_2}{\epsilon}-1\right) m
  }
< \frac{\epsilon}{\kappa_2} \sum_{m=2}^{l_\theta+1}
  \left(\frac{m}{n_\theta}\right)^{\alpha m}
\end{equation}
where $\alpha =\ln(1/\epsilon)-1$. The replacement of
$\ln(\kappa_2/\epsilon)-1$ by $\alpha$ (at the second inequality)
is valid, because $(m/n_\theta) \le 1$, and $\kappa_2 > 1$ implies
$\ln(\kappa_2/\epsilon)-1 > \alpha$. Following the proof of
Theorem~\ref{theorem:algorithm-1}, we can derive the following
upper bound of $P^\mathrm{fail}_\theta$ using
Lemma~\ref{lemma:prefactor-upper-bound-2}:
\begin{align}
&P^\mathrm{fail}_\theta \nonumber\\ 
&< \left[
   \sum_{m=2}^{\left\lfloor\frac{n_\theta}{\mathrm{e}}\right\rfloor}
   \left(\frac{2}{n_\theta}\right)^{2\alpha}\mathrm{e}^{-\beta(m-2)} +
   \sum_{m=\left\lceil\frac{n_\theta}{\mathrm{e}}\right\rceil}^{l_\theta+1}
   \mathrm{e}^{\frac{\alpha}{\mathrm{e}-1}(m-n_\theta)}
   \right] \frac{\epsilon}{\kappa_2} \nonumber\\
&< \left[
   \left(\frac{2}{n_\theta}\right)^{2\alpha}
   \sum_{m=2}^\infty\mathrm{e}^{-\beta(m-2)} +
   \sum_{m^\prime=n_\theta-l_\theta-1}^\infty
   \mathrm{e}^{-\frac{\alpha}{\mathrm{e}-1}m^\prime}
   \right]\frac{\epsilon}{\kappa_2} \nonumber\\
&= \left[
   \left(\frac{2}{n_\theta}\right)^{2\alpha}\frac{1}{1-\mathrm{e}^{-\beta}} +
   \frac{\mathrm{e}^{-\frac{\alpha}{\mathrm{e}-1}(n_\theta-l_\theta-1)}}
        {1-\mathrm{e}^{-\frac{\alpha}{\mathrm{e}-1}}}
   \right] \frac{\epsilon}{\kappa_2},
\end{align}
where the parameter $\beta$ is defined in Eq.~\eqref{eq:beta-def}.
Note that the first term in the last expression can be omitted
if$ \lfloor n_\theta / \mathrm{e} \rfloor < 2$, i.e.,
$n_\theta < 2\mathrm{e}$.

Next, we consider the failure probability of the second type: Algorithm~2
samples an optimal solution but stops before collecting all optimal solutions.
This failure probability is essentially the same as the failure probability of
Algorithm~1. Let $f_{\min}$ denote $\min_{x \in X}f(x)$, and let
$\mathcal{E}_{f_{\min}}$ be the event where the algorithm samples
an optimal solution. Under $\mathcal{E}_{f_{\min}}$, the sampling process
is governed by the probability distribution $p_{f_{\min}}$, which is
the uniform distribution on $X_{f_{\min}}$. Thus, following the proof of
Theorem~\ref{theorem:algorithm-1}, we obtain an upper bound for the failure probability
of the second type as:
\begin{align}
P^\mathrm{fail}_{f_{\min}}
&= P\left(\bigcup_{m=2}^{n_{f_{\min}}}
          \left\{
            T^{(p_{f_{\min}})}_m >
            \left\lceil m\ln\frac{m\kappa_2}{\epsilon} \right\rceil
          \right\}
          \cap \mathcal{E}_{f_{\min}}\right) \nonumber\\
&\le \sum_{m=2}^{n_{f_{\min}}}
     P\left(T^{(p_{f_{\min}})}_m >
            \left\lceil m\ln\frac{m\kappa_2}{\epsilon} \right\rceil \right) \nonumber\\
&< \left[
   \left(\frac{2}{n_{f_{\min}}}\right)^{2\alpha}
   \frac{1}{1-\mathrm{e}^{-\beta}} +
   \frac{1}{1-\mathrm{e}^{-\frac{\alpha}{\mathrm{e}-1}}}
   \right] \frac{\epsilon}{\kappa_2}.
\end{align}
Note that $n_{f_{\min}}$ in the last expression is replaced by six
in Theorem~\ref{theorem:algorithm-1}, because the first term can be neglected
for $n_{f_{\min}} < 2\mathrm{e} = 5.43\cdots$ [see the discussion after
Eq.~\eqref{eq:split-summation}]. However, we maintain the dependence on
$n_{f_{\min}}$ for subsequent discussion.

The total failure probability is bounded above by the sum of
$P^\mathrm{fail}_\theta$ across all $\theta$ in the image of $f$, i.e.,
$f[X] \coloneqq \{\theta \in \mathbb{R} \mid
 \exists x \in X\ \text{s.t.}\ f(x) = \theta\}$.
Thus, the total failure probability, denoted by $P^{\mathrm{fail}}$,
satisfies the inequality
\begin{widetext}
\begin{align}
P^{\mathrm{fail}}
&= P_{f_{\min}}^{\mathrm{fail}}
  + \sum_{\theta \in f[X] \setminus \{f_{\min}\}} P_\theta^{\mathrm{fail}}
  \nonumber \\
&< \left[
   \frac{1}{1-\mathrm{e}^{-\beta}}
   \sum_{\theta \in f[X]} \left(\frac{2}{n_\theta}\right)^{2\alpha} +
   \frac{1}{1-\mathrm{e}^{-\frac{\alpha}{\mathrm{e}-1}}}
   \left(1 + \sum_{\theta \in f[X] \setminus \{f_{\min}\}}
         \mathrm{e}^{-\frac{\alpha}{\mathrm{e}-1}(n_\theta-l_\theta-1)}\right)
   \right] \frac{\epsilon}{\kappa_2}.
\label{eq:failure-probability-theta-sum}
\end{align}
\end{widetext}

We evaluate the first summation in Eq.~\eqref{eq:failure-probability-theta-sum}.
Recall that $X_\theta$ is the set of all feasible solutions
whose cost value is less than or equal to $\theta$. In particular,
the indexed family of sets $\{X_\theta\}_{\theta \in f[X]}$
forms a strictly increasing sequence with respect to $\theta$:
for $\theta_1, \theta_2 \in f[X]$,
$\theta_1 < \theta_2 \Rightarrow X_{\theta_1} \subsetneq X_{\theta_2}$.  
Hence, $\theta_1 < \theta_2 \Rightarrow n_{\theta_1} < n_{\theta_2}$,
meaning that the sequence $\{n_\theta\}_{\theta \in f[X]}$ is
a strictly increasing sequence and has no duplicated values.
Moreover, the terms for $n_\theta < 2\mathrm{e} = 5.43\cdots$ can be
excluded from the first summation. Thus, the first summation over $\theta$ is
upper bounded by the infinite summation over $n_\theta \ge 6$ as follows:
\begin{align}
\sum_{\theta \in f[X]}
\left(\frac{2}{n_\theta}\right)^{2\alpha}
&< \sum_{n_\theta=6}^\infty \left(\frac{2}{n_\theta}\right)^{2\alpha} \nonumber\\
&= 4^\alpha \left(\zeta(2\alpha) - \sum_{k=1}^5 \frac{1}{k^{2\alpha}}\right).
\end{align}
Here, we rewrite the infinite sum in terms of the Riemann zeta function
$\zeta(s) \coloneqq \sum_{k=1}^\infty k^{-s}$. Since $\epsilon$ is set
to be less than $1/\mathrm{e}^{1.5}$, the argument $2\alpha$ exceeds 1.
This ensures the convergence of $\zeta(2\alpha)$.

Next, we evaluate the second summation in Eq.~\eqref{eq:failure-probability-theta-sum}.
Suppose $\theta_1 \in f[X] \setminus \{f_{\min}\}$. Let $\theta_0$ denote
the value in $f[X]$ immediately below $\theta_1$, i.e.,
$\theta_0 \coloneqq \max \{\theta \in f[X] \mid \theta < \theta_1\}$.
As $Y_{\theta_1}$ denotes the set of all feasible solutions whose cost value
equals $\theta_1$, $X_{\theta_1} = X_{\theta_0} \cup Y_{\theta_1}$
and $X_{\theta_0} \cap Y_{\theta_1} = \emptyset$ hold. This implies
$n_{\theta_1} - l_{\theta_1} = n_{\theta_0}$. Since the sequence
$\{n_{\theta_0}\}_{\theta_0 \in f[X]}$ is a strictly increasing
sequence of positive integers, the sequence
$\{n_{\theta_1}-l_{\theta_1}\}_{\theta_1 \in f[X]\setminus \{f_{\min}\}}$
is also a strictly increasing sequence of positive integers. Therefore,
the second summation over $\theta$ is upper bounded by the infinite summation
over positive integers $n_\theta - l_\theta$, which we denote by $k$,
as follows:
\begin{align}
1+\sum_{\theta \in f[X] \setminus \{f_{\min}\}}
\mathrm{e}^{-\frac{\alpha}{\mathrm{e}-1}(n_\theta-l_\theta-1)}
&< 1 + \sum_{k=1}^\infty \mathrm{e}^{-\frac{\alpha}{\mathrm{e}-1}(k-1)} \nonumber\\
&= \frac{2-\mathrm{e}^{-\frac{\alpha}{\mathrm{e}-1}}}
        {1-\mathrm{e}^{-\frac{\alpha}{\mathrm{e}-1}}}.
\end{align}

Finally, we derive an upper bound for the total failure probability
of Algorithm~2 as follows:
\begin{align}
&P^\mathrm{fail} \nonumber\\
&< \left[\frac{4^\alpha}{1-\mathrm{e}^{-\beta}}
         \left(\zeta(2\alpha) - \sum_{k=1}^5 \frac{1}{k^{2\alpha}}\right) +
         \frac{2-\mathrm{e}^{-\frac{\alpha}{\mathrm{e}-1}}}
              {\left(1-\mathrm{e}^{-\frac{\alpha}{\mathrm{e}-1}}\right)^2}
   \right] \frac{\epsilon}{\kappa_2}.
\end{align}
Since the expression inside the brackets on the right-hand side
equals $\kappa_2$ [Eq.~\eqref{eq:kappa2-def}], the right-hand side equals $\epsilon$.
Therefore, the failure probability of Algorithm~2 remains below $\epsilon$,
irrespective of the minimum value of $f$ and the number of optimal solutions.
\end{proof}

Note here that, while we utilize the notions of $n_\theta$ and $l_\theta$
as well as $f_\mathrm{min}$ in the derivation---quantities one could not
access unless one knew the underlying cost landscape---the final
conclusion of the proof does not rely on such inaccessible information.

This proof also clarifies that $\kappa_2$ is designed to compensate for
the increased error chances caused by the lack of information about
$f_{\min}$ as well as $n_{f_{\min}}$. In contrast, the design of
$\kappa_1$ takes into account only the failure cases due to
ignorance of $n_{f_{\min}}$ (i.e., the failure scenarios of
the second type in the above proof). Therefore, $\kappa_2$ should
be larger than $\kappa_1$. Indeed, $\kappa_2$ includes $\kappa_1$:
\begin{widetext}
\begin{align}
\kappa_2
&= \frac{4^\alpha}{1-\mathrm{e}^{-\beta}}
   \sum_{k=6}^\infty \frac{1}{k^{2\alpha}} +
   \frac{1}{1-\mathrm{e}^{-\frac{\alpha}{\mathrm{e}-1}}}
   \left(1 + \frac{1}{1-\mathrm{e}^{-\frac{\alpha}{\mathrm{e}-1}}}\right) \nonumber\\
&= \underbrace{
    \left[\frac{4^\alpha}{1-\mathrm{e}^{-\beta}}\frac{1}{6^{2\alpha}} +
          \frac{1}{1-\mathrm{e}^{-\frac{\alpha}{\mathrm{e}-1}}}\right]
   }_{= \kappa_1} +
   \left[\frac{4^\alpha}{1-\mathrm{e}^{-\beta}}
         \sum_{k=7}^\infty \frac{1}{k^{2\alpha}} +
         \frac{1}{\left(1-\mathrm{e}^{-\frac{\alpha}{\mathrm{e}-1}}\right)^2}
   \right] \nonumber\\
&= \kappa_1 +
   \left[\frac{4^\alpha}{1-\mathrm{e}^{-\beta}}
         \sum_{k=7}^\infty \frac{1}{k^{2\alpha}} +
         \frac{1}{\left(1-\mathrm{e}^{-\frac{\alpha}{\mathrm{e}-1}}\right)^2}
   \right].
\end{align}
\end{widetext}

\bibliographystyle{unsrt}

\end{document}